\documentclass{sig-alternate-2013}
\usepackage{times}
\usepackage[T1]{fontenc}
\usepackage{amsmath,amssymb}
\usepackage{units}      % gives the "\nicefrac" command
\usepackage{graphicx}   % allows the inclusion of .eps figures
\usepackage[usenames,dvipsnames,svgnames,table]{xcolor}
\usepackage{url}
\usepackage{seqsplit}
\usepackage{multirow}
\usepackage{needspace}
\usepackage{flushend}
\allowdisplaybreaks[1]

\usepackage{enumerate}
\usepackage{marginnote}
\usepackage{color}
\usepackage{url}
\usepackage[usenames,dvipsnames]{xcolor}
\usepackage[linktocpage=true,colorlinks=false]{hyperref}

\newcommand{\remove}[1]{{}}  % removes a portion of text

\newcommand{\calx}{\mathcal{X}}

\newcommand{\calz}{\mathcal{Z}}

\newcommand{\cals}{\mathcal{S}}
\newcommand{\cala}{\mathcal{A}}
\newcommand{\calb}{\mathcal{B}}

\newcommand{\calp}{\mathcal{P}}

\newcommand{\calg}{\mathcal{G}}

\newcommand{\K}{K}

\newcommand{\grid}{\calg}

\newcommand{\dmax}{{d_\infty}}

\newcommand{\hamming}{{d_h}}

\newcommand{\dprob}{{d_\calp}}
\newcommand{\bayes}{\mathbf{Bayes}}

\newcommand{\xb}{\mathbf{x}}
\newcommand{\zb}{\mathbf{z}}

\newcommand{\reals}{\mathbb{R}}

\newcommand{\edpold}[1][\epsilon]{$#1$-dif\-fer\-en\-tial privacy}

\newcommand{\geoind}{geo-indistingui\-sha\-bility}
\newcommand{\egeoind}[1][\epsilon]{$#1$-geo-in\-dis\-tin\-guisha\-bil\-ity}

\newcommand{\planar}{\mathit{PL}_\epsilon}

\newcommand{\diam}{\operatorname{diam}}

\newcommand{\lambert}{W_{\!\!-1}}

\newcommand{\smallprod}[1]{\textstyle{\prod_{#1}\:}}

\newcommand{\smallfrac}[2]{\textstyle{\frac{#1}{#2}}}
\newcommand{\smallmax}[1]{\textstyle{\max_{#1}\:}}

\newcommand{\vect}[1]{{#1}} % definition for paper on geo-indistinguishability

\newcommand{\codeComment}[1]{\emph{\color{gray}{// #1}}}

\usepackage{thmtools,thm-restate}
\declaretheorem[parent=section]{theorem}
\declaretheorem[sibling=theorem]{proposition}
\declaretheorem[parent=section]{definition}

\declaretheorem[sibling=theorem]{observation}

\newfont{\mycrnotice}{ptmr8t at 7pt}
\newfont{\myconfname}{ptmri8t at 7pt}
\let\crnotice\mycrnotice%
\let\confname\myconfname%

\permission{Permission to make digital or hard copies of all or part of this work for personal or classroom use is granted without fee provided that copies are not made or distributed for profit or commercial advantage and that copies bear this notice and the full citation on the first page. Copyrights for components of this work owned by others than the authors must be honored. Abstracting with credit is permitted. To copy otherwise, or republish, to post on servers or to redistribute to lists, requires prior specific permission and/or a fee. Request permissions from permissions@acm.org.}
\conferenceinfo{CCS'13,}{November 4--8, 2013, Berlin, Germany. \\
{\mycrnotice{Copyright is held by the owner/author(s). Publication rights licensed to ACM.}}}
\copyrightetc{ACM \the\acmcopyr}
\crdata{978-1-4503-2477-9/13/11\ ...\$15.00.\\
http://dx.doi.org/10.1145/2508859.2516735}

\title{Geo-Indistinguishability: Differential Privacy for Location-Based Systems}

\numberofauthors{2}

\author{
\alignauthor
	Miguel E. Andr\'es\\
	\affaddr{\'Ecole Polytechnique}\\
	\email{\large mandres@lix.polytechnique.fr}
\alignauthor
	Nicol\'as E. Bordenabe\\
    \affaddr{INRIA and \'Ecole Polytechnique}\\
    \email{\large nbordenabe@lix.polytechnique.fr}
\and
\alignauthor
	Konstantinos Chatzikokolakis\\
    \affaddr{CNRS and \'Ecole Polytechnique}\\
    \email{\large kostas@lix.polytechnique.fr}
\alignauthor 
	Catuscia Palamidessi\\
    \affaddr{INRIA and \'Ecole Polytechnique}\\
    \email{\large catuscia@lix.polytechnique.fr}
}

\begin{document}

\maketitle

\begin{abstract}
The growing popularity of location-based systems, allowing unknown/untrusted
servers to easily collect  huge amounts of  information regarding users'
location, has recently started raising serious privacy concerns. In this paper
we introduce \geoind{}, a formal notion of privacy for location-based systems
that protects the user's exact location, while  allowing approximate information
-- typically needed to obtain a certain desired service -- to be released.

This privacy definition formalizes the intuitive notion of protecting the user's
location within a radius $r$ with a level of privacy that depends on $r$, and
corresponds to a generalized version of the well-known concept of
\emph{differential privacy}. Furthermore, we present a mechanism for achieving
\geoind{} by adding controlled random noise to the user's location.

We describe how to use our mechanism to enhance LBS applications with
geo-indistinguishability guarantees without compromising the quality of the
application results. Finally, we compare state-of-the-art mechanisms from the
literature with ours. It turns out that, among all mechanisms independent of the
prior, our mechanism offers the best privacy guarantees.
\end{abstract}

\vspace{-3ex}

\vspace{-3ex}
%\begin{keywords} 
%Geolocation.
%Privacy Technologies.
%Differential privacy.
%Location-based services.
%Data sanitation.
%Random perturbation techniques.
%Planar laplacian distribution.
%\end{keywords}

\vspace{-3ex}
\section{Introduction}

In recent years, the increasing availability of location information about
individuals has led
to a growing use of systems that
record and process location data, generally
referred to as ``location-\vspace{30ex} based systems''. Such systems include (a) Location Based
Services (LBSs), in which a user obtains, typically in real-time, a service
related to his current location,
%as well as
and
(b) location-data mining algorithms,
used to determine
%, among others,
points of interest and traffic patterns.
%, and geographical distributions of diseases.

The use of LBSs, in particular, has been significantly increased by the
growing popularity of mobile devices equipped with GPS chips, in combination
with the increasing availability of wireless data connections. A resent
study in the US shows that in $2012$, $46\%$ of the adult population of the country owns a
smartphone and, furthermore, that $74\%$ of those owners use LBSs
\cite{PewInternet:12:Survey}. Examples of LBSs include mapping applications
(e.g., Google Maps%, Bing Maps
), Points of Interest (POI) retrieval (e.g.,
AroundMe%, Localscope
), coupon/discount pro\-viders (e.g., GroupOn%, Yowza
),
GPS navigation (e.g., TomTom), and loca\-tion-aware social networks
(e.g., Foursquare%, OkCupid
).
%LBS users typically submit their
%location to a service provider returns in order to obtain a certain benefit, e.g.,
%information about POI in the area around them.

While location-based systems have demonstrated to provide enormous benefits to
individuals and society, the growing exposure of users' location information
raises important privacy issues.
% that are, unfortunately, often overlooked.
% On the one side, 
First of all, location information itself may be considered % by individuals
as sensitive. 
% More importantly, location data 
Furthermore, it can be easily linked to a variety
of other information that an individual usually wishes to protect: by
collecting and processing accurate location data on
a regular basis, it is possible to infer an individual's home or work
location, sexual preferences, political views, religious inclinations, etc. In
its extreme form, monitoring and control of an individual's location has been
even described as a form of slavery \cite{Dobson:03:TSM}.

Several notions of privacy for location-based systems have
been proposed in the literature.
% many of them being variations of the $k$-anonymity concept, together with techniques to achieve these privacy guarantees. 
In Section~\ref{sec:existing-notions} we give an overview of such notions, 
% existing notions of location privacy, 
and we discuss their shortcomings in relation to our motivating LBS application. 
Aiming at addressing these shortcomings, we propose  a \emph{formal privacy
definition} for LBSs, as well as a randomized technique that allows
a user  to disclose \emph{enough location information} to obtain the desired service,
while satisfying the aforementioned privacy notion. Our proposal is  based on 
a generalization of \emph{differential privacy} \cite{Dwork:06:ICALP} developed 
in \cite{Chatzikokolakis:13:PETS}. Like differential privacy, our
notion and technique abstract from the side information of the adversary, such as
any prior probabilistic knowledge about the user's actual location.

\looseness -1
As a running example, we consider a user located in Paris who wishes to
query an LBS provider for nearby restaurants in a private way, i.e., by
disclosing some approximate information $z$ instead of his exact location $x$.
%Note that, in contrast to various works in the literature, we assume
%that the user is interested in hiding his \emph{location}, not his
%\emph{identity}. In fact, the user might be willing to disclose
%his identity to the provider, in order to obtain personalized recommendations, or to
%participate in a social network. 
A crucial question  is:
what kind of privacy guarantee can the user expect in this scenario? 
%On the one hand, he does not expect to reveal his exact location but, on the other
%hand, he wishes to obtain a service tailored to it. Thus, the user's
%requirement is that, by obtaining $z$, the provider should be able to infer $x$
%\emph{approximately} but not \emph{accurately}.
To formalize this notion, we consider the level of privacy \emph{within a radius}.
%We fix a circle of radius $r$ centered at the user's location
%$x$, and reason about the user's level of privacy within this radius. 
We say that the user enjoys \emph{$\ell$-privacy within $r$} if,
any
two locations at distance at most $r$ produce observations with ``similar''
distributions, where the ``level of similarity'' 
depends on $\ell$. The
idea is that $\ell$ represents the user's \emph{level} of privacy for that
radius: the smaller $\ell$ is, the higher is the privacy. 
%(as it gets
%harder for the provider to detect the user's location among the locations
%within this circle).

In order to allow the  LBS to provide a useful service,
we require that the (inverse of the) level of privacy  $\ell$ depend on the radius $r$. In
particular, we require that it is proportional to $r$, which brings us to our 
%(still informal) 
definition of \emph{\geoind{}}:
\begin{quote}
	A mechanism satisfies \egeoind{}
	iff for any radius $r>0$, the user enjoys $\epsilon r$-privacy within $r$.
\end{quote}
This definition implies that the user is protected within any radius $r$, but
with a level $\ell =\epsilon r$ that increases with the distance. Within a
short radius, for instance $r\!=\!1$ km, $\ell$ is small, guaranteeing that
the provider cannot infer the user's location within, say, the 7th
arrondissement of Paris. Farther away from the user, for instance for $r = 1000$ km,
$\ell$ becomes large, allowing the LBS provider to infer that with high
probability the user is located in Paris instead of, say,  London. 
Figure~\ref{fig:paris} illustrates the idea of privacy levels decreasing with the radius.

 \begin{figure}[tb]%
		\centering
		\vspace{3 mm}
		\includegraphics[width=0.7\columnwidth]{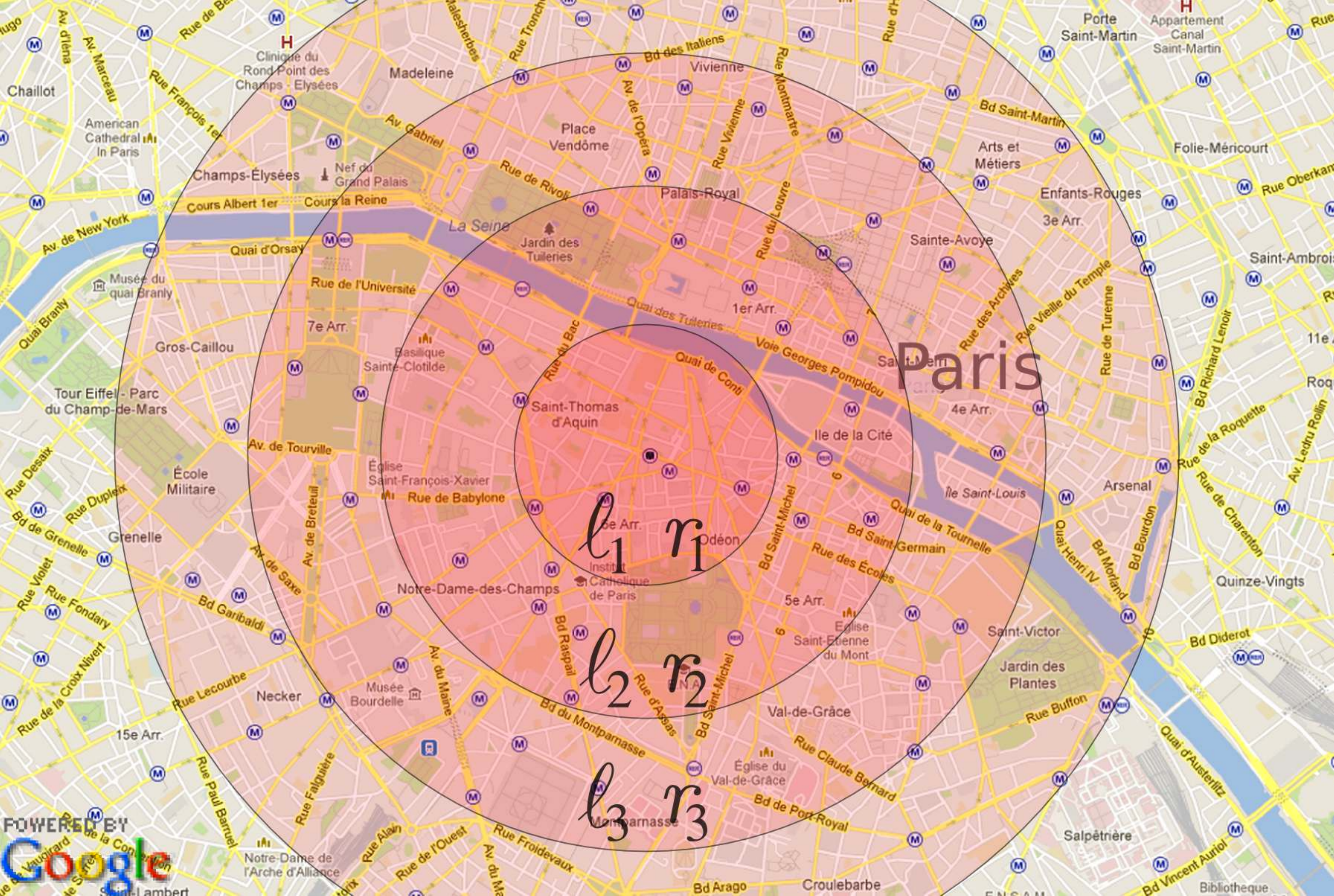}%\\[2ex]
		\caption{Geo-indistinguishability:  privacy varying with $r$.}%
		\label{fig:paris}%
		%\vspace{-0.4cm}
\end{figure}

We develop a mechanism to achieve \geoind{} by perturbating the user's
location $x$. The inspiration  comes from one of the most
popular approaches for differential privacy, namely the
Laplacian noise. 
We adopt a specific planar version of the Laplace distribution,  allowing to 
draw points in  a \emph{geo-indistinguishable} way;
moreover, we are able to do so efficiently, via a transformation to polar coordinates.
However, as standard (digital) applications
require a finite representation of locations, it is necessary to add a
discretization step. 
Such operation jeopardizes the privacy guarantees, for reasons similar to the  rounding effects of finite-precision operations  \cite{Mironov:12:CCS}. 
We show how to preserve  geo-indistinguishability, at the price of a degradation of the privacy level, and 
how to adjust the privacy parameters  in order to obtain a desired level of privacy.

We then describe how to use our mechanism to enhance LBS applications with geo-indistinguishability guarantees. Our proposal results in highly configurable LBS applications, both in terms of privacy and accuracy (a notion of utility/quality-of-service for LBS applications providing privacy via location perturbation techniques). Enhanced LBS applications require extra bandwidth consumption in order to provide both privacy and accuracy guarantees, thus we study how the different configurations affect the bandwidth overhead using the Google Places API~\cite{GooglePlacesApi} as reference to measure bandwidth consumption. Our experiments showed that the bandwidth overhead necessary to enhance LBS applications with very high levels of privacy and accuracy is not-prohibitive and, in most cases, negligible for modern applications.

Finally, we compare our mechanism with other ones in the literature, using the
privacy metric proposed in \cite{Shokri:12:CCS}. It turns
our that our mechanism offers the best privacy guarantees, for the same utility,
among all those which do not depend on the prior knowledge of the adversary. The
advantages of the independence from the prior are obvious: first, the mechanism
is designed once and for all (i.e. it does not need to be recomputed every time
the adversary changes, it works also in simultaneous presence of different
adversaries, etc.). Second,   and even more important, it is applicable  also
when we do not know the prior. 

\paragraph{Contribution}
This paper contributes to the state-of-the-art as follows:
\begin{itemize}
\item
We show that our generalized notion of differential privacy \cite{Chatzikokolakis:13:PETS}, instantiated with the Euclidean metric, can be naturally applied to location privacy, and we discuss the privacy guarantees that this definition provides. (Location privacy was only briefly mentioned in \cite{Chatzikokolakis:13:PETS} as a possible application.) 
\item 
We also extend it to location traces, using the $\dmax$ metric, and show how privacy degrades when traces become longer.
\item 
We propose a mechanism to efficiently draw noise from a planar Laplace distribution, which is not trivial. Laplacians on general metric spaces were briefly discussed in \cite{Chatzikokolakis:13:PETS}, but no efficient method to draw from them was given. Furthermore, we cope with the crucial problems of discretization and truncation, which have been shown to pose significant threats to mechanism implementations  \cite{Mironov:12:CCS}.
\item 
We describe how to use our mechanism to enhance LBS applications with geo-indistinguishability guarantees.
\item 
We compare our mechanism to a state-of-the-art mechanism from the literature \cite{Shokri:12:CCS} as well as a simple cloaking mechanism, obtaining favorable results.
\end{itemize}

\paragraph{Road Map} In Section 2 we discuss notions of location privacy from the
literature and point out  their weaknesses and strengths. In Section 3 we
formalize the notion of \geoind{} in three equivalent ways. We
then proceed to describe
a mechanism that provides \geoind{} in Section 4. In Section 5
we show how to enhance LBS applications with geo-indistinguishability guarantees. In Section 6 we compare the privacy guarantees of our methods with those of two 
other methods from the literature. 
Section 7 discusses related
work and Section 8 concludes. 

The interested reader can find the proofs  in the report version of this paper~\cite{Andres:12:CoRR}, which is available online.

%All proofs are in the appendix.

\section{Existing Notions of Privacy}
\label{sec:existing-notions}

In this section, we examine various notions of location privacy from the literature, as
well as techniques to achieve them. We consider the motivating example from the
introduction, of a user in Paris wishing to find nearby
restaurants with good reviews. To achieve this goal, he uses a handheld device
(e.g.. a smartphone) to query a public LBS provider.
However, the user expects his location to be kept private: informally speaking,
the information sent to the provider should not allow him to accurately infer
the user's location. Our goal is to provide a \emph{formal} notion of privacy
that adequately captures the user's expected privacy.
From the point of view of the employed mechanism, we require a
technique that can be performed in real-time by a handheld device,
without the need of any trusted anonymization party.

\paragraph{Expected Distance Error}

Expectation of distance error \cite{Shokri:11:SP,Shokri:12:CCS,Hoh:05:SecureComm} is a natural way to quantify the
privacy offered by a location-obfuscation mechanism. Intuitively, it reflects the degree of
accuracy by which an adversary can guess the real location of the user by observing the obfuscated
location, and using the side-information available to him. 
%In general, location-protection mechanisms based on this 
%definition aim at maximizing the adversary's expected error while guaranteeing some
%minimum user-imposed quality of service.

%We can find in the literature several works relying on this privacy notion.
There are several works relying on this notion. 
In \cite{Hoh:05:SecureComm}, a perturbation mechanism
is used to confuse the attacker by crossing paths of individual users, rendering 
the task of tracking individual paths challenging. In \cite{Shokri:12:CCS}, 
an optimal location-obfuscation mechanism (i.e., achieving maximum level of
privacy for the user)
is obtained by solving a linear program
in which the contraints are determined by the quality of service and by the user's profile. 
%given constraints about the bound in the quality of service and the
%frequency with which the user is located in the different regions considered (known in this case as the \emph{user's profile}).

It is worth noting that this privacy notion 
and the obfuscation mechanisms based on it
are explicitly defined in terms of the adversary's side information. 
%This implies that location-obfuscation mechanisms based on this notion
%assume that the attacker have some
%particular kind of side-information (for instance, past location traces of the user).
In contrast, our notion of \geoind{} abstracts from
the attacker's prior knowledge, and is therefore suitable for scenarios
where the prior is unknown, or the same mechanism must be used for multiple users.
A detailed comparison with the mechanism of \cite{Shokri:12:CCS}
is provided in Section~\ref{sec:comparison}.
%This assumption implies that the definition is only satisfied for this
%limited class of adversaries.
%In our case, we aim for a privacy definition that is independent form
%the attacker's knowledge 
%This means that it is necessary to make assumptions about the information of the attacker in order to state that the location-obfuscation mechanism satisfies a certain level of privacy.
%This means that, for a particular mechanism, it is not possible to state that an adversary will have an expected distance error of $E$ without making assumptions about the side-information he may have.
%In \cite{shokri 2}, side-information is even used to derive the mechanism itself.
%We will, then, aim for a notion of privacy that is not defined in terms of the attacker's side-information.

\paragraph{$k$-anonymity}

The notion of $k$-anonymity is the most widely used definition of privacy for
location-based systems in the literature. Many systems in this category
\cite{Gruteser:03:MobiSys, Gedik:05:ICDCS, Mokbel:06:VLDB} aim
at protecting the user's \emph{identity}, requiring that the attacker cannot
infer which user is executing the query, among a set of $k$ different users.
Such systems are outside the scope of our problem, since we are
interested in protecting the user's \emph{location}.

On the other hand, $k$-anonymity has also been used to protect the user's
location (sometimes called $l$-diversity in this context), requiring that it is
indistinguishable among a set of $k$ points (often required to share some
semantic property). One way to achieve this is through the use of \emph{dummy
locations} \cite{Kido:05:ICDE,Shankar:09:UbiComp}. This technique
involves generating $k-1$ properly selected dummy points, and performing $k$
queries to the service provider, using the real and dummy locations. Another
method for achieving $k$-anonymity is through \emph{cloaking}
\cite{Bamba:08:WWW,Duckham:05:Pervasive,Xue:09:LoCa}. This involves creating a
cloaking region that includes $k$ points sharing some property of interest, and
then querying the service provider for this cloaking region.

Even when side knowledge does not explicitly appear in the definition of $k$-anonymity,
a system cannot be proven to satisfy this notion unless assumptions are made about the
attacker's side information. 
%The main drawback of $k$-anonymity-based approaches in general is that a system
%cannot be proved to satisfy this notion unless assumptions are made about the
%attacker's side information. 
For example, dummy locations are only useful
if they look equally likely to be the real location from the point of view of
the attacker. Any side information that allows to rule out any of those
points, as having low probability of being the real location, would immediately
violate the definition.

Counter-measures are often employed to avoid this issue: for instance,
\cite{Kido:05:ICDE}  takes into account concepts such as ubiquity, congestion
and uniformity for generating dummy points, in an effort to make them
look realistic. Similarly, \cite{Xue:09:LoCa} takes into account the user's side
information to construct a cloaking region. Such counter-measures
have their own drawbacks: first, they complicate the employed techniques, also
requiring additional data to be taken into account
(for instance, precise information about the environment or the location of nearby users),
making their application in
real-time by a handheld device challenging. Moreover, the attacker's actual
side information might simply be inconsistent with the assumptions being
made.

As a result, notions that abstract from the attacker's side information,
such as differential privacy, have been growing in popularity in recent years,
compared to $k$-anonymity-based approaches.

\paragraph{Differential Privacy}

Differential Privacy \cite{Dwork:06:ICALP} is a notion of privacy from the
area of statistical databases. Its goal is to protect an individual's data while
publishing aggregate information about the database. Differential privacy
requires that modifying a single user's data should have a negligible effect on
the query outcome. More precisely, it requires that the probability that a query
returns a value $v$ when applied to a database $D$, compared to the probability
to report the same value when applied to an \emph{adjacent} database $D'$ --
meaning that $D,D'$ differ in the value of a single individual -- should be
within a bound of $e^\epsilon$.
% (where $\epsilon$ is a parameter quantifying the
%level of privacy). 
A typical way to achieve this notion is to add controlled
random noise to the query output, for example drawn from a Laplace distribution.
An advantage of this notion is that a mechanism can be shown to be
differentially private independently from any side information that the
attacker might possess.

Differential privacy has also been used in the context of location privacy. In
\cite{Machanavajjhala:08:ICDE}, it is shown that a synthetic data
generation technique can be used to publish statistical information about
commuting patterns in a differentially private way. In \cite{Ho:11:GIS},
a quadtree spatial decomposition technique is used to ensure differential
privacy in a database with location pattern mining capabilities.

As shown in the aforementioned works, differential privacy can be successfully
applied in cases where \emph{aggregate} information about several users is published.
On the other hand, the nature of this notion makes it poorly suitable for
applications in which only a single individual is involved, such as our motivating
scenario. The secret in this case is the location of a single user. Thus,
differential privacy would require that any change in that location should have
negligible effect on the published output, making it impossible to communicate
any useful information to the service provider.

To overcome this issue, Dewri \cite{Dewri:12:TMC} proposes a mix of differential
privacy and $k$-anonymity, by fixing an anonymity set of $k$ locations and requiring
that the probability to report the same obfuscated location $z$ from any of
these $k$ locations should be similar (up to $e^\epsilon$). This property is
achieved by adding Laplace noise to each Cartesian coordinate independently.
There are however two problems with this definition: first, the choice of
the anonymity set crucially affects the resulting privacy; outside this
set no privacy is guaranteed at all. Second, the property itself is rather weak;
reporting the geometric median (or any deterministic function) of the $k$
locations would satisfy the same definition, although the privacy guarantee
would be substantially lower than using Laplace noise.

Nevertheless, Dewri's intuition of using Laplace noise\footnote{%
	The planar Laplace distribution that we use in our work, however, 
	is different from the distribution obtained by adding Laplace noise to each Cartesian coordinate, 
	and has better differential privacy properties 
	(c.f. Section~\ref{sec:amechanism}).
}
for location privacy is valid, and \cite{Dewri:12:TMC} provides extensive experimental analysis supporting
this claim. Our notion of \geoind{} provides the formal background for
justifying the use of Laplace noise, while avoiding the need to fix an anonymity
set by using the generalized variant of differential privacy from
\cite{Chatzikokolakis:13:PETS}.

\paragraph{Other location-privacy metrics}

%There are also other location-privacy definitions that can be found in the literature, usually specific to
%some particular obfuscation mechanism. 
%
\cite{Cheng:06:PET} proposes a location cloaking mechanism, 
and focuses on the evaluation of Location-based Range Queries. The degree of privacy is measured
by the size of the cloak (also called \emph{uncertainty region}), and by the coverage of sensitive regions,
which is the ratio between the area of the cloak and the area of the regions inside the cloak
that the user considers to be sensitive. 
In order to deal with the side-information that the attacker may have, 
ad-hoc solutions are proposed, like patching cloaks to enlarge the uncertainty region or delaying requests. 
Both solutions may cause a degradation in the quality of service.

In \cite{Ardagna:07:DAS}, the real location of the user is assumed to have some level of inaccuracy, due to the specific sensing technology or to the environmental conditions. Different obfuscation techniques are
then used  to increase this inaccuracy in order
to achieve a certain level of privacy. 
This level of privacy is 
%computed as (the opposite of) the \emph{relevance} of the location measurement. 
%Relevance is 
defined as the ratio between the accuracy before and after the application of the obfuscation techniques.

Similar to the case of $k$-anonymity, both privacy metrics mentioned above make implicit assumptions about the adversary's side information. This may imply a violation of the privacy definition in a scenario where the adversary has some knowledge % (maybe probabilistic) 
about the user's real location.

\paragraph{Transformation-based approaches}

A number of approaches for location privacy are radically different from the ones
mentioned so far. Instead of cloaking the user's location, they aim at making it
completely invisible to the service provider. This is achieved by transforming
all data to a different space, usually employing cryptographic techniques, so
that they can be mapped back to spatial information only by the user
\cite{Khoshgozaran:07:SSTD,Ghinita:08:SIGMOD}. The data stored in the
provider, as well as the location send by the user are encrypted. Then, using
techniques from \emph{private information retrieval}, the provider can return
information about the encrypted location, without ever discovering which actual
location it corresponds to.

A drawback of these techniques is that they are computationally demanding,
making it difficult to implement them in a handheld device. Moreover, they
require the provider's data to be encrypted, making it impossible to use
existing providers, such as Google Maps, which have access to the real data.

\section{Geo-Indistinguishability}\label{sec:definitions}

In this section we formalize our notion of \geoind{}. As already discussed in the
introduction, the main idea behind this notion is that,
for any radius $r>0$, the user enjoys $\epsilon r$-privacy within $r$,
i.e. the level of privacy is proportional to the radius.
Note that the parameter $\epsilon$ corresponds to the level of privacy at one
unit of distance. For the user, a simple way to specify his privacy requirements
is by a tuple $(\ell,r)$, where $r$ is the radius he is mostly concerned with
and $\ell$ is the privacy level he wishes \emph{for that radius}. In this case,
it is sufficient to require \egeoind{} for $\epsilon = \ell/r$; this will ensure
a level of privacy $\ell$ within $r$, and a proportionally selected level for
all other radii.

So far we kept the discussion on an informal level by avoiding to explicitly
define what $\ell$-privacy within $r$ means. In the remaining of this section we
give a formal definition, as well as two characterizations which clarify the
privacy guarantees provided by \geoind{}.

\paragraph{Probabilistic model}

We first introduce a simple model used in the rest of the paper. We start with a
set $\calx$ of \emph{points of interest}, typically the user's possible
locations. Moreover, let $\calz$ be a set of possible \emph{reported values},
which in general can be arbitrary, allowing to report obfuscated locations,
cloaking regions, sets of locations, etc. However, to simplify the discussion,
we sometimes consider $\calz$ to also contain spatial points, assuming an
operational scenario of a user located at $x\in\calx$ and communicating to the
attacker a randomly selected location $z\in\calz$ (e.g. an obfuscated point).

Probabilities come into place in two ways. First, the attacker might have side
information about the user's location, knowing, for example, that he is likely
to be visiting the Eiffel Tower, while unlikely to be swimming in the Seine
river. The attacker's side information can be modeled by a \emph{prior}
distribution $\pi$ on $\calx$, where $\pi(x)$ is the probability assigned to the
location $x$.

Second, the selection of a reported value in $\calz$ is itself probabilistic; for
instance, $z$ can be obtained by adding random noise to the actual location $x$
(a technique used in Section~\ref{sec:mechanism}). A \emph{mechanism} $\K$ is a
probabilistic function for selecting a reported value; i.e. $\K$
is a function assigning to each location $x\in\calx$ a probability distribution on $\calz$,
where $\K(x)(Z)$ is the probability that the reported point
belongs to the set $Z \subseteq \calz$, when the user's location is
$x$.\footnote{For simplicity we assume distributions on $\calx$ to be discrete,
but allow those on $\calz$ to be
continuous (c.f.
Section~\ref{sec:mechanism}). All sets to which probability is assigned
are implicitly assumed to be measurable.}
Starting from $\pi$ and using Bayes' rule, each observation $Z\subseteq\calz$ of a
mechanism $\K$ induces a \emph{posterior} distribution
$\sigma=\bayes(\pi,K,Z)$ on $\calx$, defined as
$
	\sigma(x) = \smallfrac{ K(x)(Z)\pi(x) }{ \sum_{x'} K(x')(Z)\pi(x')}
$.

We define the \emph{multiplicative distance} between two distributions $\sigma_1,\sigma_2$
on some set $\cals$ as
$
	\dprob(\sigma_1,\sigma_2) \allowbreak= \allowbreak \sup_{S\subseteq\cals}| \ln \smallfrac{\sigma_1(S)}{\sigma_2(S)}|
$,
with the convention that $|\ln\smallfrac{\sigma_1(S)}{\sigma_2(S)}|=0$ if both
$\sigma_1(S),\sigma_2(S)$ are zero and $\infty$ if only one of them is zero.

%Together, $P_X$ and $\K$ induce a \emph{joint}
%probability distribution $P$ for $X,Z$, as $P(x,S) = P_X(x)
%\K(x)(S)$. Note that, by
%construction, $P(x) = P_X(x)$ and $P(S|x) = \K(x)(S)$.

\subsection{Definition}
\label{sec:third-approach}

We are now ready to state our definition of \geoind. Intuitively, a privacy
requirement is a constraint on the distributions $\K(x),\K(x')$ produced by two
different points $x,x'$. Let $d(\cdot,\cdot)$ denote the Euclidean metric.
Enjoying $\ell$-privacy within $r$ means that for any $x,x'$ s.t.
$d(x,x')\allowbreak\le\allowbreak r$, the distance $\dprob(\K(x),\allowbreak\K(x'))$
between the corresponding distributions should be at most $l$. Then, requiring
$\epsilon r$-privacy for all radii $r$,
forces the two distributions to be similar for locations close to each other,
while relaxing the constraint for those far away from each other, allowing a
service provider to distinguish points in Paris from those in London.
%The complete definition can be written as follows:

\begin{definition}[geo-indistinguishability]
	\label{def:geoind}
	A mechanism $\K$ satisfies \egeoind{} iff for all $x,x'$:
	\[
		\dprob(\K(x),\K(x')) \le \epsilon d(x,x')
	\]
\end{definition}
Equivalently, the definition can be formulated as
$
	\K(x)(Z) \le $ $\allowbreak e^{\epsilon d(x,x')} \allowbreak \K(x')(Z)
$
for all $x,x'\in\calx,Z\subseteq\calz$.
Note that for all points $x'$ within a radius $r$ from $x$, the definition
forces the corresponding distributions to be at most $\epsilon r$ distant.

The above definition is very similar to the one of differential privacy, which
requires $\dprob(\K(x),\K(x')) \le \epsilon \hamming(x,x')$, where $\hamming$ is
the Hamming distance between databases $x,x'$, i.e. the number of individuals in
which they differ. In fact, \geoind{} is an instance of a generalized variant of
differential privacy, using an arbitrary metric between secrets. This
generalized formulation has been known for some time: for instance,
\cite{Reed:10:ICFP} uses it to perform a compositional analysis of standard
differential privacy for functional programs, while \cite{Dwork:12:ITCS} uses
metrics between individuals to define
``fairness'' in classification. On the other hand, the usefulness of using
different metrics to achieve different privacy goals and the semantics of the
privacy definition obtained by different metrics have only recently started to
be studied \cite{Chatzikokolakis:13:PETS}. This paper focuses on location-based
systems and is, to our knowledge, the first work considering privacy under the
Euclidean metric, which is a natural choice for spatial data.

Note that in our scenario, using the Hamming metric of standard differential
privacy -- which aims at completely protecting the value of an individual --
would be too strong, since the only information is the location of a single
individual. Nevertheless, we are not interested in completely hiding the user's
location, since some approximate information needs to be revealed in order to
obtain the required service. Hence, using a privacy level that depends on the
Euclidean distance between locations is a natural choice.

%Finally, we can show that the three definitions of geo-indistingui\-shability
%given in this section are simply different ways of expressing the same privacy requirement.
%\begin{restatable}{theorem}{thmcoincide}
%	Geo-indistinguishability-I, II, III coincide.
%	\label{thm:coincide}
%\end{restatable}

\paragraph{A note on the unit of measurement}
It is worth noting that, since $\epsilon$ corresponds to the privacy level for one unit of
distance, it is affected by the unit in which distances are measured. For
instance, assume that $\epsilon=0.1$ and distances are measured in meters. The
level of privacy for points one kilometer away is $1000\epsilon$, hence changing
the unit to kilometers requires to set $\epsilon = 100$ in order for the
definition to remain unaffected. In other words, if $r$ is a physical quantity
expressed in some unit of measurement, then $\epsilon$ has to be expressed in
the inverse unit.
%In this paper we omit the unit since the choice is orthogonal to our goals.

\subsection{Characterizations}

In this section we state two characterizations of \geoind,
obtained from the corresponding results of \cite{Chatzikokolakis:13:PETS} (for
general metrics), which provide intuitive interpretations of the privacy
guarantees offered by \geoind.

\paragraph{Adversary's conclusions under hiding}

The first characterization uses the concept of a \emph{hiding function}
$\phi:\calx\to\calx$. The idea is that $\phi$ can be applied to the user's
actual location before the mechanism $\K$, so that the latter has only access to
a hidden version $\phi(x)$, instead of the real location $x$. A mechanism $\K$
with hiding applied is simply the composition $\K\circ\phi$. Intuitively, a
location remains private if, regardless of his side knowledge (captured by his
prior distribution), an adversary draws the same conclusions (captured by his
posterior distribution), regardless of whether hiding has been applied or not.
However, if $\phi$ replaces locations in Paris with those in London, then clearly
the adversary's conclusions will be greatly affected. Hence, we require that the
effect on the conclusions depends on the maximum distance $d(\phi)
\allowbreak=\sup_{x\in\calx}d(x,\phi(x))$ between the real and hidden location.

%\begin{mydef}{Geo-indistinguishability-I}
%	A mechanism satisfies \egeoind{} iff for all priors $P_X$
%	and all $S\subseteq\calz$:\footnote{Note that for
%	the sake of readability, we express the definitions in terms of fractions.
%	To avoid issues with zero probabilities, we can write all definitions in flat
%	form, i.e. $P(x|S)P(x') \le e^{\epsilon r} P(x'|S) P(x)$.}
%	\[
%		\frac{P(x|S)}{P(x'|S)} \le e^{\epsilon r} \frac{P(x)}{P(x')}
%		\qquad \forall r > 0 \;\forall x,x' : d(x,x') \le r
%	\]
%\end{mydef}
\begin{theorem}
	\label{thm:characterization_alt}
	A mechanism $\K$ satisfies \egeoind{}
	iff
	for all $\phi:\calx\to\calx$,
	all priors $\pi$ on $\calx$,
	and all $Z \subseteq\calz$:
	\begin{align*}
		\dprob(\sigma_1, \sigma_2) &\le 2\epsilon d(\phi)
			&\textrm{where}\qquad
			\sigma_1 &= \bayes(\pi,K,Z) \\
			&&\sigma_2 &= \bayes(\pi,K\circ\phi,Z)
	\end{align*}
\end{theorem}

Note that this is a natural adaptation of a well-known interpretation of
standard differential privacy, stating that the attacker's conclusions are
similar, regardless of his side knowledge, and regardless of whether an
individual's real value has been used in the query or not. This corresponds to a
hiding function $\phi$ removing the value of an individual.

Note also that the above characterization compares two \emph{posterior}
distributions. Both $\sigma_1,\sigma_2$ can be substantially different than
the initial knowledge $\pi$, which means that an adversary does learn
some information about the user's location.

\paragraph{Knowledge of an informed attacker}
\label{sec:second-approach}

A different approach is to measure how much the adversary learns about the
user's location, by comparing his prior and posterior distributions. However,
since some information is allowed to be revealed by design, these distributions
can be far apart. Still, we can consider an \emph{informed} adversary who
already knows that the user is located within a set $N\subseteq\calx$. Let $d(N)
= \sup_{x,x'\in N}d(x,x')$ be the maximum distance between points in $x$.
Intuitively, the user's location remains private if, regardless of his prior
knowledge within $N$, the knowledge obtained by such an informed adversary
should be limited by a factor depending on $d(N)$. This means that if $d(N)$ is
small, i.e. the adversary already knows the location with some accuracy, then
the information that he obtains is also small, meaning that he cannot improve
his accuracy. Denoting by $\pi_{|N}$ the distribution obtained from $\pi$ by
restricting to $N$ (i.e. $\pi_{|N}(x) = \pi(x|N)$), we obtain the following
characterization:

\begin{theorem}
	\label{thm:characterization}
	A mechanism $\K$ satisfies \egeoind{}
	iff
	for all $N\subseteq\calx$,
	all priors $\pi$ on $\calx$,
	and all $Z \subseteq\calz$:
	\begin{align*}
		\dprob(\pi_{|N}, \sigma_{|N}) &\le \epsilon d(N)
			&\textrm{where}\qquad
			\sigma &= \bayes(\pi,K,Z)
	\end{align*}
\end{theorem}

Note that this is a natural adaptation of a well-known interpretation of
standard differential privacy, stating that in informed adversary who already
knows all values except individual's $i$, gains no extra knowledge from the
reported answer, regardless of side knowledge about $i$'s value
\cite{Dwork:06:TCC}.

\paragraph{Abstracting from side information}
A major difference of \geoind, compared to similar approaches from the
literature, is that it abstracts from the side information available to the
adversary, i.e. from the prior distribution. This is a
subtle issue, and often a source of confusion, thus we would like to clarify what
``abstracting from the prior'' means. The goal of a privacy definition
is to restrict the information \emph{leakage} caused by the observation. Note
that the lack of leakage does not mean that the user's location cannot be inferred (it
could be inferred by the prior alone), but instead that the adversary's
knowledge does not increase \emph{due to the observation}.

However, in the context of LBSs, no privacy definition can ensure a small
leakage under any prior, and at the same time allow reasonable utility.
Consider, for instance, an attacker who knows that the user is located at some
airport, but not which one. The attacker's prior knowledge is very limited,
still any useful LBS query should reveal at least the user's city, from which
the exact location (i.e. the city's airport) can be inferred. Clearly,
due to the side information, the leakage caused by the observation is high.

So, since we cannot eliminate leakage under any prior, how can we give a
reasonable privacy definition without restricting to a particular one? First, we
give a formulation (Definition~\ref{def:geoind}) which does not involve the
prior at all, allowing to verify it without knowing the prior. At the same time,
we give two characterizations which
explicitly quantify over all priors, shedding light on how the prior affects
the privacy guarantees.

Finally, we should point out that differential privacy abstracts from the prior
in exactly the same way. Contrary to what is sometimes believed, the user's
value is \emph{not protected} under any prior information. Recalling the
well-known example from~\cite{Dwork:06:ICALP}, if the adversary knows that Terry
Gross is two inches shorter than the average Lithuanian woman, then he can
accurately infer the height, even if the average is release in a differentially
private way (in fact no useful mechanism can prevent this leakage). Differential
privacy does ensure that her risk is the same whether she participates
in the database or not, but this might me misleading: it does not imply the lack
of leakage, only that it will happen anyway, whether she participates
or not!
%On the other hand, it can be shown that if an adversary already knows
%the value of all people but Terry Gross, then the leakage caused by releasing
%the query is small. This is the exact analogy of knowing $B_r(x)$ in the case of
%\geoind{}.

%\pagebreak
\subsection{Protecting location sets}
\label{sec:multiple-locations}
So far, we have assumed that the user has a single location that he wishes to
communicate to a service provider in a private way (typically his current
location). In practice, however, it is
common for a user to have multiple points of interest, for instance a set of
past locations or a set of locations he frequently visits. In this case, the user might
wish to communicate to the provider some information that depends on all
points; this could be either the whole set of points itself, or some aggregate
information, for instance their centroid.
As in the case of a single location, privacy is still a requirement; the provider is allowed
to obtain only approximate information about the locations, their exact value
should be kept private. In this section, we discuss how \egeoind{} extends to
the case where the secret is a tuple of points $\xb = (x_1,\ldots,x_n)$.

Similarly to the case of a single point, the notion of distance is crucial for
our definition. We define the distance between two tuples of points $\xb =
(x_1,\ldots,x_n), \xb' = (x_1',\ldots,x_n')$ as:
\[
	\dmax(\xb,\xb') = \smallmax{i} d(x_i, x_i')
\]
Intuitively, the choice of metric follows the idea of reasoning
within a radius $r$: when $\dmax(\xb,\xb')\le r$, it means that all
$x_i,x_i'$ are within distance $r$ from each other.
All definitions and results of this section can be then directly applied to the case of
multiple points, by using $\dmax$ as the underlying metric. Enjoying $\ell$-privacy
within a radius $r$ means that two tuples at most $r$ away from each other,
should produce distributions at most $\epsilon r$ apart.

%the observation can help the attacker to infer $\xb$
%among all tuples at distance $r$ (i.e. tuples having all points at distance $r$
%from the corresponding points of $\xb$), by a factor of at most $e^l$. All three
%definitions of \geoind{} remain the same, the only change being
%the set of secrets and the distance between them.

\paragraph{Reporting the whole set}

A natural question then to ask is how we can obfuscate a tuple
of points, by independently applying an existing mechanism $\K_0$ to each
individual point, and report the obfuscated tuple. Starting from a tuple
$\xb=(x_1,\ldots,x_n)$, we independently apply $\K_0$ to each $x_i$ obtaining
a reported point $z_i$, and then report the tuple $\zb=(z_1,\ldots,z_n)$. Thus,
the probability that the combined mechanism $\K$ reports $\zb$, starting from
$\xb$, is the product of the probabilities to obtain each point $z_i$, starting
from the corresponding point $x_i$, i.e. $\K(\xb)(\zb) = \prod_i
\K_0(x_i)(z_i)$.
%\footnote{For simplicity we consider probabilities of points
%	here; a formal treatment of continuous mechanism would require to consider
%sets.}

The next question is what level of privacy does $\K$ satisfy. For simplicity,
consider a tuple of only two points $(x_1,x_2)$, and assume that
$\K_0$ satisfies \egeoind{}. At first look, one might expect the
combined mechanism $\K$ to also satisfy \egeoind{}, however
this is not the case. The problem is that the two points might be
\emph{correlated}, thus an observation about $x_1$ will reveal information
about $x_2$ and vice versa. Consider, for instance, the extreme case in which $x_1 = x_2$.
Having two observations about the same point reduces the level of privacy, thus
we cannot expect the combined mechanism to provide the same level of privacy.

\looseness -1
Still, 
	if $\K_0$ satisfies \egeoind{}, then
	$\K$ can be shown to satisfy \egeoind[n\epsilon],
	%\label{thm:composition}
i.e.
a level of privacy that scales linearly with $n$.
%\begin{restatable}{theorem}{thmcomposition}
%	If $\K_0$ satisfies \egeoind{}, then
%	$\K$ satisfies \egeoind[n\epsilon].
%	\label{thm:composition}
%\end{restatable}
%
%Note that this issue is similar to the problem of composing queries in standard
%differential privacy. If the outcome of multiple queries is randomized by
%adding independent noise to each answer, then $\epsilon$ scales linearly with
%the number of queries. The reason is exactly that the answers
%are correlated, since they come from the same database.
Due to this scalability issue, the technique of independently applying a mechanism
to each point is only useful when the number of points is small. Still, this is
sufficient for some applications, such as the case study of
Section~\ref{sec:location-based}. Note, however, that this technique is by no means
the best we can hope for: similarly to standard differential privacy
\cite{Blum:08:STOC,Roth:10:STOC}, better results could be achieved by adding
noise to the whole tuple $\xb$, instead of each individual point. We
believe that using such techniques we can achieve \geoind{} for a large number
of locations with reasonable noise, leading to practical mechanisms for
highly mobile applications. We have already started exploring this
direction of future work.

\paragraph{Reporting an aggregate location}

Another interesting case is when we need to report some aggregate information
obtained by $\xb$, for instance the centroid of the tuple. In general
we might need to report the result of a query $f:\calx^n\to\calx$.
Similarly to the case of standard differential privacy, we can compute the
real answer $f(\xb)$ and the add noise by applying a mechanism $\K$ to it.
If $f$ is $\Delta$-sensitive wrt $d,\dmax$, meaning that $d(f(\xb),f(\xb'))\le \Delta
\dmax(\xb,\xb')$ for all $\xb,\xb'$, and $\K$ satisfies \geoind{},
then the composed mechanism $\K\circ f$ can be shown
to satisfy \egeoind[\Delta\epsilon].

Note that when dealing with aggregate data, standard differential privacy
becomes a viable option. However, one needs to also examine the loss of utility
caused by the added noise. This highly depends on the application: differential
privacy is suitable for publishing aggregate queries with \emph{low
sensitivity}, meaning that changes in a single individual have a relatively
small effect on the outcome. On the other hand, location information often has
high sensitivity. A trivial example is the case where we want to publish the
complete tuple of points. But sensitivity can be high even for aggregate
information: consider the case of publishing the centroid of 5 users located
anywhere in the world. Modifying a single user can hugely affect their centroid,
thus achieving differential privacy would require so much noise that the result
would be useless. For \geoind{}, on the other hand, one needs to consider the
distance between points when computing the sensitivity. In the case of the
centroid, a small (in terms of distance) change in the tuple has a small effect
on the result, thus \geoind{} can be achieved with much less noise.

\section{A mechanism to achieve geo-in\-distinguishability}\label{sec:mechanism}

In this section we present  a method to generate noise so to satisfy  \geoind{}.
We model the location domain as a discrete\footnote{ For applications with digital interface  the domain of interest is discrete, since the representation of the coordinates of the
points is  necessarily 
finite.} Cartesian plane  with the standard notion of Euclidean distance.
%Namely, given two points with coordinates
%$(s_1,t_1)$ and $(s_2,t_2)$ respectively, their distance is defined as
%$d((s_1,t_1),(s_2,t_2))=\sqrt{(s_2-s_1)^2+(t_2-t_1)^2}$.
This model can be considered a good approximation of the Earth surface when the area of interest is not too large.

%However, it is not easy to devise an efficient  mechanism for \geoind{}
%that generates noise directly on a discrete plane (we will come back to this point in Section~\ref{sec:discrete}).
%We therefore consider a different approach:

\begin{enumerate}[(a)]
\item First, we define  a  mechanism to achieve geo-indistinguishability in the ideal case of the continuous plane.
%For each actual location,  this mechanism should generate a random point in a way that satisfies \geoind{}  on $\reals^2$.
\item \looseness -1
Then, we discretized the mechanism by remapping each  point generated according to (a) to the closest point in the discrete domain.
%\end{enumerate}
%Furthermore, we  may want to consider only a limited area. For instance if we are in a island, we may wish to report only   locations in the land, not in the sea. Thus we may want to apply a third step:
%\begin{enumerate}[(c)]
\item Finally, we  truncate the mechanism, so to report only points within the limits of the area of interest.
\end{enumerate}

\subsection{A mechanism for the continuous plane}\label{sec:amechanism}
\looseness -1
Following the above plan, we  start by defining a  mechanism for  geo-indistinguishability on the continuous plane.
%This will constitute the basis of our method.
The idea is that whenever the actual location is $\vect{x}_0\in\reals^2$,  we report, instead, a point $\vect{x}\in\reals^2$ generated randomly according to the noise function.
The latter needs  to be such  that the probabilities of reporting a point in a certain (infinitesimal) area around $\vect{x}$,
when the actual locations are $\vect{x}_0$ and $\vect{x}'_0$ respectively,  differs at most by a multiplicative factor $e^{-\epsilon \,{d(\vect{x}_0,\vect{x}'_0)}}$.

We can achieve this property by requiring that the probability of generating a point in the area around $\vect{x}$  decreases exponentially with the distance
from the actual location $\vect{x}_0$.
In a linear space this is exactly the behavior of the Laplace distribution, whose
probability density function (pdf) is
$
 \nicefrac{\epsilon}{2}\, e^{- \epsilon\,|x-\mu|}
$.
%(where $\mu,\epsilon$ are parameters).
This distribution has been used in the literature to add noise to query results on statistical databases, with $\mu$ set to be the actual answer,
and it can be shown to satisfy  \edpold{} \cite{Dwork:11:CACM}.
%Figure~\ref{fig:linear laplacians} illustrates the idea.
%\begin{figure}[tb]%
% 		\centering
% 		\includegraphics[width=0.9\columnwidth]{figures/linearLaplacians1.pdf}%\\[-2ex]
%			\caption{Two linear Laplacians (with $\epsilon=\ln 2$) centered at $1$ and $2$.
%			The ratio between the two curves is at most $e^\epsilon = 2$ everywhere.}%
%			\label{fig:linear laplacians}%
%			\vspace{-7pt}
%\end{figure}
%Of course we cannot use the standard Laplace distribution for our purposes, because it is defined on the line, while we need a distribution  defined on the plane. Furthermore we need to use the (Euclidean) planar distance $d(\vect{x},\vect{\mu})$ instead of the liner distance  $|x-\mu|$.
%Intuitively,
%%however, just
%replacing  $|x-\mu|$ by $d(\vect{x},\vect{\mu})$ in the Laplace's pdf
%results in  a natural extension of the Laplacian  from one to two dimensions. 

There are two possible definitions of 
Laplace distribution on high\-er dimensions (multivariate Laplacians).
 %In general \emph{multivariate} means that it involves $k\geq 1$ random variables. The particular cases of $k=1$ and $k=2$  are called \emph{univariate} and \emph{bivariate} respectively.
The first one, investigated in %\cite{Lange:93:JCGS,Arslan:10:StatPap}.
\cite{Lange:93:JCGS}, and used also in  \cite{Dwork:06:TCC},  is obtained from the standard Laplacian by replacing  $|x-\mu|$ with $d(\vect{x},\vect{\mu})$. 
%We call  \emph{planar Laplacian} such extension.
The second way consists in generating each Cartesian coordinate independently, according to a linear Laplacian. 
For reasons that will become clear in the next paragraph, we adopt the first approach. 
%\newpage

\paragraph{The probability density function}
Given  the parameter   $\epsilon\in\reals^+$, and the actual location $\vect{x}_0\in\reals^2$,   the pdf  of our noise mechanism, on any other point $\vect{x}\in\reals^2$,  is:
\begin{equation}\label{eqn:planar laplacian}
D_\epsilon (\vect{x}_0) (\vect{x}) = \frac{\epsilon^2}{2\,\pi} \,  e^{-\epsilon \,{d(\vect{x}_0,\vect{x})}}
\end{equation}
where $\nicefrac{\epsilon^2}{2\,\pi} $ is a normalization factor.
%Using a transformation in polar coordinates it is possible to show that the integral of  this  function over the whole $\reals^2$ gives $1$, which means that it is indeed the pdf of a probability distribution.
%In Appendix~\ref{app:normalization factor} we show how to compute the normalization factor.
 We call this function  \emph{planar Laplacian centered at} $\vect{x}_0$. 
The corresponding distribution is illustrated in Figure~\ref{fig:planar laplacians}.
It is possible to show that (i) the projection of a planar Laplacian on any vertical plane passing by the center gives a (scaled) linear Laplacian,  %(Figure~\ref{fig:linear laplacians}).
%Or equivalently, the planar laplacian can be obtained by rotating a linear laplacian around its center, and multiplying by a normalization factor.
%In %the report version of this paper  %\cite{techrep} (Appendix B)
%Appendix~\ref{app:geo-indistinguishability}
and  (ii) the corresponding mechanism  satisfies $\epsilon$-\geoind{}.
These two properties would not be  satisfied by the second approach to the multivariate Laplacian.

 \begin{figure}[t]%
		\centering
		\includegraphics[width=0.8\columnwidth,height=0.5\columnwidth]{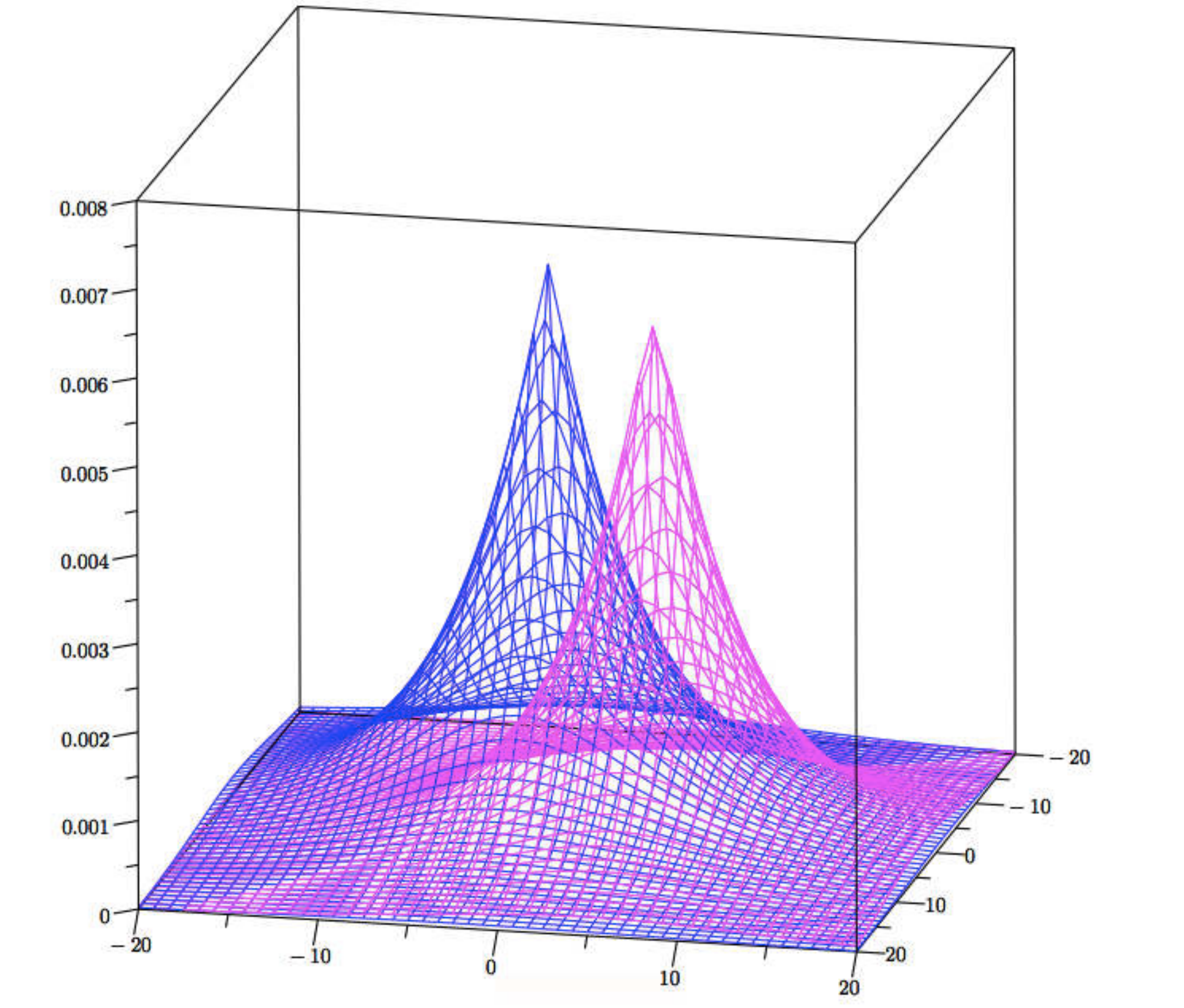}%
		\caption{The pdf of two planar Laplacians, centered at $(-2,-4)$ and at $(5,3)$ respectively, with $\epsilon=1/5$.%
		%The distance between the  centers is $7\,\sqrt{2}$, and the ratio between the  curves is at most $e^{7/5\,\sqrt{2}} \approx 7.24$ everywhere.
		}
		\label{fig:planar laplacians}%
		\vspace{-3pt}
\end{figure}

\paragraph{Drawing a random point}\label{sec:drawing polar}
We illustrate now how to draw a random point from the   pdf  defined in \eqref{eqn:planar laplacian}.
First of all, we note that  the pdf of the planar Laplacian depends only on the distance from $\vect{x}_0$.
It will be convenient,  therefore, to switch to a system of polar coordinates with origin in $\vect{x}_0$.
% Intuitively, in this way the pdf will depend only on one variable, thus simplifying the drawing procedure.
%So, given the pdf in \eqref{eqn:planar laplacian}, we consider the transformation into a system of polar coordinates $(r,\theta)$  where $r$ is the  radius and $\theta$ is the angle. 
A point $\vect{x}$ will be represented as a point $(r,\theta)$, where $r$ is the distance of $\vect{x}$ from $\vect{x}_0$, and $\theta$ is the angle that the line $\vect{x}\,\vect{x}_0$ forms with respect to the horizontal axis of the Cartesian system.
Following the standard transformation formula,  the pdf of the \emph{polar  Laplacian} centered at the origin ($\vect{x}_0$) is:
\begin{equation}\label{eqn:planar laplacian polar}
D_\epsilon (r,\theta) = \frac{\epsilon^2}{2\,\pi} \, r\, e^{-\epsilon \,r}
\end{equation}

We note now that the polar Laplacian defined above  enjoys a property that is very convenient for drawing in an efficient way:
\emph{the two  random variables  that represent the radius and the angle are  independent}.
Namely, the pdf can be expressed as the product of the two marginals.
In fact, let us denote these  two random variables by $R$  (the radius) and $\Theta$ (the angle). The two marginals are:
\[
\begin{array}{rcl}
%\label{eqn:Gamma}
D_{\epsilon,R} (r) &= &   \int_0^{2\pi} D_\epsilon (r,\theta) \,  d\theta \, = \,  \epsilon^2\, r\, e^{-\epsilon \,r} \\[2ex]
D_{\epsilon,\Theta} (\theta) &= & \int_0^{\infty} D_\epsilon(r,\theta) \, d r  \, = \,   \frac{1}{2\,\pi}
\end{array}
\]
Hence we have
$
D_\epsilon  (r,\theta) = D_{\epsilon,R} (r) \; D_{\epsilon,\Theta} (\theta)
$.
Note that $D_{\epsilon,R} (r)$ corresponds to the pdf of the \emph{gamma distribution} with shape $2$ and scale $1/\epsilon$. 
%Figure~\ref{fig:Gamma}(a) shows the graph of this function for various values of  $\epsilon$.
%
%\begin{figure}[tb]%
%		\centering
%		%\includegraphics[width=0.9\columnwidth, height=0.6\columnwidth]{figures/gamma.pdf}%
%		\includegraphics[width=\columnwidth, height=0.35\columnwidth]{figures/charts.pdf}\\%
%		\caption{Gamma distribution: pdf and cdf for various $\epsilon$.}%
%		\label{fig:Gamma}%
%		\vspace{-3pt}
%\end{figure}

%It may come as a surprise that this graph differs significantly from those in Figures~\ref{fig:linear laplacians} and \ref{fig:planar laplacians}, and in particular, that it does not have its maximum in the origin.
%Remember, however, that the graph in Figure ~\ref{fig:Gamma}(a) represents a pdf \emph{in polar coordinates}. More precisely, $D_{\epsilon,R} (r)$  represents the probability that the random point is located in the circular crown  centered at the origin and delimited by  $r$ and $r+dr$. The area of this crown is proportional to $r$, hence when $r$ is close to $0$ also the probability is close to $0$. As $r$ increases the probability increases, until the factor $e^{-\epsilon\, r}$ takes over. For  $r$ approaching infinity, the factor $e^{-\epsilon\, r}$ approaches $0$, and dominates over $r$, hence the probability approaches $0$ again.
%

Thanks to the fact that  $R$ and $\Theta$ are independent, in order to  draw a point $(r,\theta)$ from $D_\epsilon  (r,\theta)$ it is sufficient to draw separately $r$ and $\theta$ from $D_{\epsilon,R} (r) $ and $D_{\epsilon,\Theta} (\theta) $ respectively.

Since $D_{\epsilon,\Theta} (\theta) $ is constant, drawing $\theta$ is easy: it is sufficient to generate $\theta$ as a random  number in the interval $[0,2\pi)$ with uniform distribution.

We now show how to draw $r$. Following    standard lines, we consider the cumulative distribution function (cdf) $C_{\epsilon}(r)$:
\begin{equation*}
	C_{\epsilon}(r) \,=\, \displaystyle \int_0^r     D_{\epsilon,R} (\rho) d\rho \,=\,
	%&=& \displaystyle \left.(-\frac{\epsilon\,\rho}{2\,\pi}-\frac{1}{2\,\pi})\, e^{-\epsilon\,\rho}\, \right|_0^r \\[4ex]
	 \displaystyle 1-(1+\epsilon\,r)\, e^{-\epsilon\,r}
\end{equation*}
Intuitively, $C_{\epsilon}(r)$ 
%(Figure~\ref{fig:Gamma}(b)) 
represents the probability that the radius of the random point  falls between $0$ and $r$.
Finally,  we generate a random number $p$ with uniform probability in the interval $[0,1)$,  and  we set $r = C_{\epsilon}^{-1}(p)$.
Note that
\[
	C_\epsilon^{-1}(p) = - \smallfrac{1}{\epsilon}\big(\lambert(\smallfrac{p-1}{e})+1\big)
\]
where $\lambert$ is the Lambert W function (the $-1$ branch), which can be computed efficiently and is
implemented in several numerical libraries (MATLAB, Maple, GSL, \ldots).

%\begin{figure}[tb]%
%		\centering
%		\includegraphics[width=0.9\columnwidth]{figures/cumulativegamma.pdf}%
%		\caption{Plot of $C_{\epsilon}(r)$ for various values of $\epsilon$.}%
%		\label{fig:cumulativeGamma}%
%\end{figure}

%Given a ``universal'' Cartesian reference system and  the actual location $\vect{x}_0= (s,t)$ in this system, if we could work in the ``ideal'' continuous plane, then we would just need to generate the noise  $(r,\theta)$
%as specified above, and then reports the point $\vect{x}= (s + r \cos\theta,t+r\sin\theta)$.
%In practice however there is always some discretization involved, because (a) computers have finite precision, and (b) (more important)   the coordinates of the ``universal reference system'' will have a finite representation, typically using only a few decimal digits. The discretization of our method, and its properties, constitute the subject  of next  section.

\subsection{Discretization}\label{sec:discrete}
%In practical applications locations are typically represented by means of discrete coordinates. For instance, latitude and longitude up to some decimal of precision.
%Thus
We discuss now how to  approximate the Laplace mechanism on a grid $\grid$ of discrete Cartesian coordinates. Let us recall the points (a) and (b) of the plan, in light of the development so far: 
Given  the actual location  $\vect{x}_0$,   report the point $\vect{x}$ in $\grid$ obtained as follows:
\begin{enumerate}[(a)]
\item\label{item:polar point} first,     draw  a point $(r,\theta)$ following the method in  Figure~\ref{fig:method for the noise},

%Figure~\ref{fig:method for the noise} summarizes our method.
 \begin{figure}[t]%
 \begin{tabular}{l}
 \hline \vspace{-3mm}\\
{\bf Drawing a  point $(r,\theta)$ from the polar Laplacian}\\
 \hline \vspace{-3mm}\\
1. \ draw $\theta$ uniformly in $[0, 2\pi)$\\
2.  \  draw   $p$ uniformly in  $[0, 1)$ and   set $r = C_{\epsilon}^{-1}(p)$\\
\hline
\end{tabular}\\[-2ex]
\bigskip
\caption{Method to generate Laplacian noise.}
\label{fig:method for the noise}
\vspace{-5pt}
\end{figure}

\item\label{item:grid point} then,    remap  $(r,\theta)$  to the closest point  $\vect{x}$ on $\grid$.
 \end{enumerate}

 We will denote by ${\K}_{\epsilon}:  \grid\to\calp(\grid)$ the above mechanism.  In summary, ${\K}_{\epsilon}(\vect{x}_0)(\vect{x})$ represents the probability
 of reporting the point $\vect{x}$ when the actual point is $\vect{x}_0$.

It is not obvious that the discretization preserves \geoind{}, due to the following problem:
 In principle,  each point $\vect{x}$ in $\grid$ should gather the probability of the set of points for which
 $\vect{x}$ is the closest point in $\grid$, namely %the rectangular region
 \[\mathit{R}(\vect{x}) =\{\vect{y}\in\reals^2\, |\, \forall \vect{x}'\in {\cal G}. \; d(\vect{y},\vect{x}') \leq d(\vect{y},\vect{x}')\}\]
However, due to the finite precision of the machine,  the noise generated according to  (\ref{item:polar point}) is already discretized in accordance with the polar system. Let $\cal W$ denote the  discrete set of points actually generated in (\ref{item:polar point}).
Each of those points $(r,\theta)$  is drawn with the probability of the area between $r$, $r+\delta_r$, $\theta$ and $\theta+\delta_\theta$, where  $\delta_r$ and  $\delta_\theta$ denote the precision of the machine in representing the radius and the angle respectively.
Hence, step (\ref{item:grid point}) generates a point  $\vect{x}$ in $\grid$ with the probability of the set
$
\mathit{R}_{\cal W}(\vect{x}) =
%\{\vect{y}\in{\cal W}\, |\, \forall \vect{x}'\in {\cal G}. \; d(\vect{y},\vect{x}') \leq d(\vect{y},\vect{x}')\}  =
\mathit{R}(\vect{x}) \cap {\cal W}
$.
This introduces some  irregularity in the mechanism, because  the %scaly
region associated to  $\mathit{R}_{\cal W}(\vect{x}) $
has a different   shape and area depending on the position of $\vect{x}$ relatively to $\vect{x}_0$.
%Figure~\ref{fig:polar to grid} illustrates the situation.
The situation is illustrated in Figure~\ref{fig:polar to grid}  with $R_0=R_{\cal W}(\vect{x}_0)$ and $R_1=R_{\cal W}(\vect{x}_1)$.

%The Cartesian grid constituted by blue horizontal and vertical lines   represents $\grid$. The  polar grid constituted by black circles and radial lines   represent $\cal W$.  The two dashed rectangles around the points  $\vect{x}_0$ and  $\vect{x}_1$   represent $R(\vect{x}_0)$ and  $R(\vect{x}_1)$.
%%\footnote{We have chosen  $\vect{x}_0$ to coincide with the origin of the polar system, but this is not essential: it is only to   emphasize the difference in the remapping.}.
%The regions $R_0$ and $R_1$  colored  in grey and magenta correspond to $R_{\cal W}(\vect{x}_0)$ and $R_{\cal W}(\vect{x}_1)$ respectively. Note that  $R_0$ and $R_1$ have different shapes and areas, for instance $R_0$ is larger than $R_1$.
%%Also note that, the farther we are from the origin, the less uniform the areas of the scaly regions become.

%In the next paragraph we show that, despite of the above problem, we can still obtain \geoind{},
%at the price of a degradation in  the privacy parameter $\epsilon$.
\begin{figure}[tb]%
		\centering
		\includegraphics[width=0.5\columnwidth]{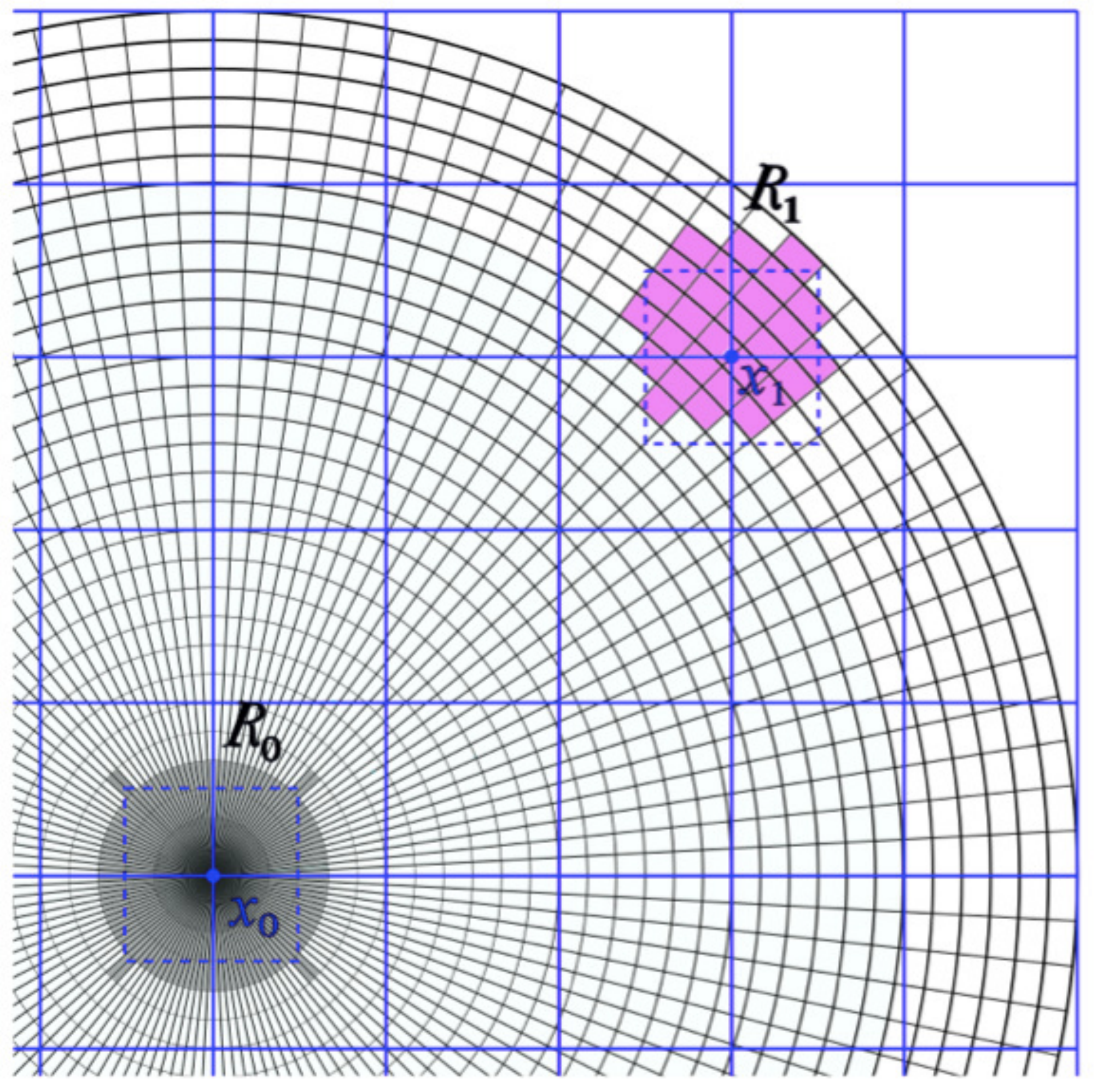}%
		\caption{Remapping the points in polar coordinates   to  points in the grid.}%
		\label{fig:polar to grid}%
		\vspace{-5pt}
\end{figure}

\paragraph{Geo-indistinguishability of the discretized mechanism}
We   now   analyze the privacy guarantees provided by our discretized mechanism.
We  show  that the discretization  preserves \geoind{}, at the price of a degradation of  the privacy parameter $\epsilon$.

For the sake of generality we do not require the step units along the two dimensions of $\grid$ to be equal.
We will call  them \emph{grid units}, and will denote by $u$ and $v$  the smaller and the larger unit, respectively.
 %\[
%u=\max\{u_x,u_y\}
%\]
We recall that $\delta_\theta$ and  $\delta_r$ denote the precision of the machine in representing $\theta$ and $r$, respectively.
We assume that $\delta_r \le r_{\max}\delta_\theta$.
The following theorem  % \cite{techrep}   (Appendix~C),
 states the \geoind{} guarantees provided by our mechanism:
$\K_{\epsilon'}$ satisfies  $\epsilon$-\geoind{}, within a range $r_{\max}$,
provided that  $\epsilon'$  is chosen in a suitable way
that depends on $\epsilon$,  on  the length of
the step units  of $\grid$, and on the precision of the machine.

\begin{restatable}{theorem}{thmgeodp}
	\label{theo:geo-dp}
	Assume $r_{\max} <   \nicefrac{u}{ \delta_\theta}$, and let $q = \nicefrac{u}{r_{\max} \delta_\theta}$. Let $\epsilon,\epsilon'\in\reals^+$ such that
	\[\epsilon' + \frac{1}{u} \ln \, \frac{q + 2\, e^{\epsilon'u}}{q - 2\, e^{\epsilon'u}}\leq \epsilon\]
	Then
	 ${\K}_{\epsilon'}$ provides $\epsilon$-\geoind{}  within  the range of
	 $r_{\max}$. Namely,
	 if $d(\vect{x}_0,\vect{x}),d(\vect{x}'_0,\vect{x})\leq r_{\max}$ then:
	\[
	{\K}_{\epsilon'} (\vect{x}_0)(\vect{x})\leq e^{\epsilon\, d(\vect{x}_0,\vect{x}_0')}{{\K}_{\epsilon'}(\vect{x}_0')(\vect{x})}.
	\]
\end{restatable}

The difference between $\epsilon'$ and $\epsilon$
represents the additional noise needed  to compensate the effect of discretization.
Note that  $r_{\max}$, which determines the area in which  $\epsilon$-\geoind{} is guaranteed, must be chosen in such a way that
$q > 2\, e^{\epsilon'u}$. Furthermore there is a trade-off between $\epsilon'$ and $r_{\max}$:
If we want $\epsilon'$ to be close to $\epsilon$ then we need $q$ to be large. Depending on the precision, this may or may not imply a serious limit on $r_{\max}$. Vice versa, if we want   $r_{\max}$  to be large then, depending on the precision,   $\epsilon'$ may need to be significantly  smaller than
$\epsilon$, and furthermore we may have a constraint on  the minimum possible value for $\epsilon$, which means that we may not have the
possibility of achieving an arbitrary level of  \geoind{}.
%In Appendix~\ref{app:relation} we show that the amount of extra noise can vary a lot depending on the precision of the machine;
%for double precision the extra noise is not excessive.

Figure~\ref{fig:epsilon_prime}  shows how the  additional noise varies depending on the precision of the machine. 
In this figure, $r_{\max}$ is set to be $10^2$ km, and we consider the cases of double precision ($16$ significant digits, i.e., $\delta_\theta = 10^{-16}$), single precision ($7$ significant digits), and an intermediate precision of $9$ significant digits.  Note that  
with double precision the additional noise is negligible.
\begin{figure}[tb]%
		\centering
		\includegraphics[width=0.9\columnwidth]{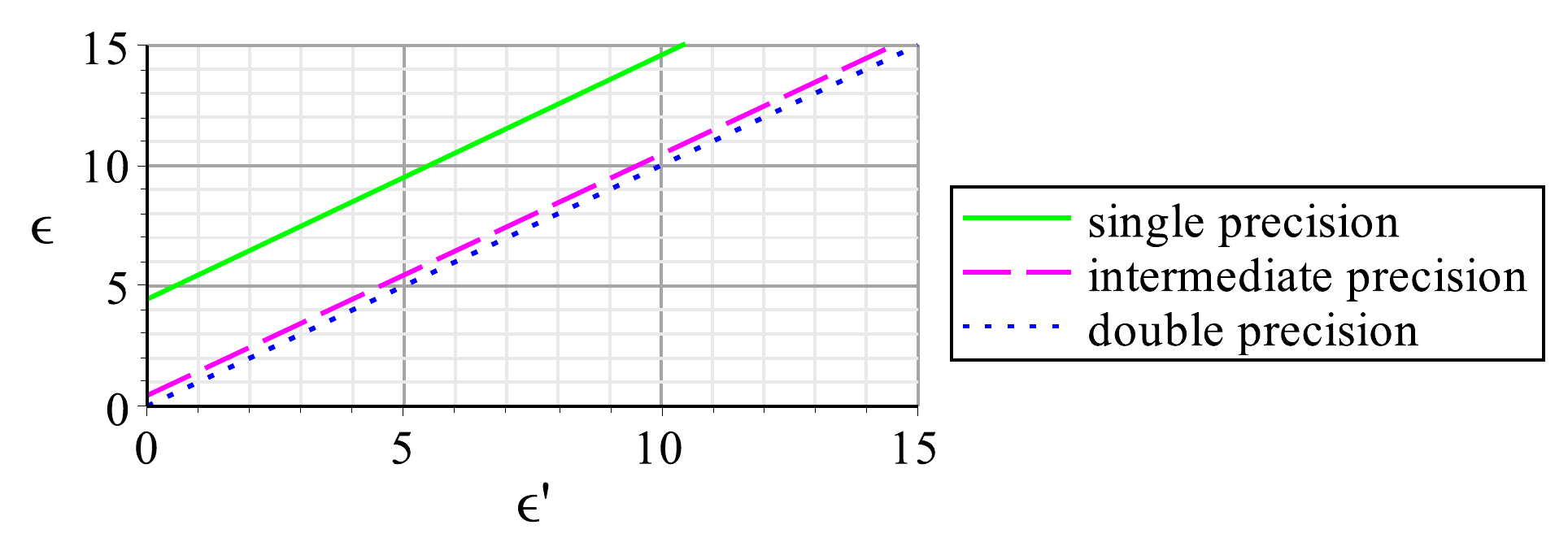}\\[-1ex]%
		\caption{The relation between $\epsilon$ and  $\epsilon'$ for  $r_{\max} = 10^2$ km.  }%
		\label{fig:epsilon_prime}%
		%\vspace{-5pt}
\end{figure}

Note that in Theorem~\ref{theo:geo-dp} the restriction about $ r_{\max}$ is crucial. Namely,  $\epsilon$-\geoind{} does not hold for arbitrary distances for any finite $\epsilon$.
Intuitively, this is because  the step units of $\cal W$ (see Figure~\ref{fig:polar to grid})
become larger with the distance $r$ from $\vect{x}_0$. The step units of  $\grid$, on the other hand, remain the same.
When  the steps in $\cal W$ become larger than those of $\grid$,   some  $\vect{x}$'s have an empty   $\mathit{R}_{\cal W}(\vect{x})$. Therefore   when  $\vect{x}$ is far away from $\vect{x}_0$ its probability may or may not be $0$, depending on  the position of $x_0$ in $\grid$, which means that   \geoind{} cannot be satisfied.

%On the other hand, the restriction on $ r_{\max}$ is not a strong limitation, because the  distribution decreases exponentially with
%$r$, and $r_{\max}$ is usually large, hence the points with  distance $r> r_{\max}$   have   negligible probability.
%Also the assumption that $u\geq \delta_r q$ is not restrictive for the kind of applications we are targeting, where $u$ is several orders of magnitude larger than $\delta_r$.

%\pagebreak
\subsection{Truncation}\label{sec:truncation}
The Laplace mechanisms described in the previous sections have the potential to
generate points everywhere in the plane, which causes several issues: first,
digital applications have finite memory, hence these mechanisms are not
implementable. Second, the discretized mechanism of Section~\ref{sec:discrete}
satisfies \geoind{} only within a certain range, not on the full plane.
Finally, in practical applications we are anyway interested in locations within
a finite region (the earth itself is finite), hence it is desirable that
the reported location  lies within that region. For the above reasons, we
propose a truncated variant of the discretized mechanism which generates points
only within a specified region and fully satisfies \geoind. The full mechanism
(with discretization and truncation) is referred to as ``Planar Laplace
mechanism'' and denoted by $\planar$.

We assume a finite set $\cala\subset\reals^2$ of admissible locations, with diameter
$\diam({\cala})$ (maximum distance between points in $\cala$). This set
is \emph{fixed}, i.e. it does not depend on the actual location $x_0$.
Our truncated mechanism  $\planar:  \cala\to\calp(\cala\cap\grid)$ works
like the discretized Laplacian of the previous section, with the difference that the
point generated in step (\ref{item:polar point}) is remapped to the closest
point in $\cala\cap \grid$. The complete mechanism is shown in
Figure~\ref{fig:algo_point}; note that step 1 assumes that
$\diam(\cala) < \nicefrac{u}{ \delta_\theta}$, otherwise no such $\epsilon'$ exists.

%Let us denote by  ${\K}_{\epsilon'}^T:  \cala\to\calp(\cala\cap\grid)$  the truncated variant of the mechanism ${\K}_{\epsilon'}$ described in previous section.
%The type is: ${\K}_{\epsilon'}^T:  \cala\to\calp(\cala\cap\grid)$ and the drawing is  described by the  following procedure. Given the actual location $\vect{x}_0\in \cala$:
%\begin{enumerate}[(a)]
%\item first,  draw  a point $(r,\theta)$  from the polar Laplacian centered at $\vect{x}_0$, as explained in previous section,
% \end{enumerate}
%\begin{enumerate}[(b$'$)]
% \item then,  remap  $(r,\theta)$  to the closest point  $\vect{x}$ on $\cala\cap\grid $.
% \end{enumerate}

% \begin{figure}[t]%
%		\centering
%		\includegraphics[width=0.5\columnwidth]{figures/circle.pdf}%
%		\caption{Probability of $\vect{x}$ in the truncated discrete laplacian.}%
%		\label{fig:circle}%
%\end{figure}

%Intuitively,  ${\K}^T$ behaves like ${\K}$ except when the region $R(\vect{x})$ is on the border of $\cala$.
%%, see Figure~\ref{fig:circle}.
%In this  case,
%the probability on $\vect{x}$ is given not only by the probability of the points in  $R_{\cal W}(\vect{x})$, but also by the probability of  the part of the cone
%%$C$
%determined by  $\vect{o}$ and $R(\vect{x})$ which lies  outside
%$\cala$.
%
%We now show that the Planar Laplace mechanism satisfies \geoind{} on all $\cala$.
\begin{restatable}{theorem}{thmgeotruncateddp}
	\label{theo:geo-truncated-dp}
	 $\planar$ satisfies  $\epsilon$-\geoind{}.
%	namely
%	\[
%	{\K}_{\epsilon'}^T (\vect{x}_0)(\vect{x})\leq e^{\epsilon\, d(\vect{x}_0,\vect{x}_0')}{{\K}_{\epsilon'}^T (\vect{x}_0')(\vect{x})}
%	\quad
%	\mbox{for every $ \vect{x}_0, \vect{x}'_0\in \cala$}
%	\]
\end{restatable}

%In the following we generally assume $\cala= r_{\max}$.

\section{Enhancing LBS{\sc s} with Privacy}
\label{sec:location-based}

%\begin{itemize}
%\item{motivation}
%\item{setting}
%\item{solution}
%\item{pros-cons}

%\item{usefulness}
%\item{further problems}
%\end{itemize}
%

%The growing use of mobile devices equipped with GPS chips in
%combination with the increasing availability of wireless and GSM
%connection has significantly increased the use of LBSs. A resent
%study in the US shows that $46\%$ of the adult population of the
%country owns a smart-phone and, furthermore, that $74\%$ of those
%owners use LBSs \cite{PewInternet:12:Survey}. Examples of LBSs include mapping
%applications (e.g., Google Maps and Bing Maps),
%Points of Interest (POI) retrieval (e.g., AroundMe and Localscope), coupon/discount providers (e.g., GroupOn and Yowza), GPS navigation (e.g., TomTom and Google
%Maps), and Location-Aware social networks (e.g.,
%Foursquare and OkCupid).
%
%Users invoking an LBS typically submit their location
%in order to obtain a certain benefit, eg, information about POI in the area around them.
%Although LBSs have proved to offer important benefits for a
%variety of applications,
%the privacy exposure of users' location information is undeniable
%and, unfortunately, often overlooked. LBS providers can collect
%accurate location information about users and, potentially,
%process it enabling them to infer sensitive information such as
%users' home location, work location, sexual preferences, political
%views, and religious inclinations.

In this section we present a case study of our privacy mechanism in the context
of LBSs. 
%In particular we show how to enhance LBS applications with privacy
%guarantees while still providing a high quality service to their users.
%
%\subsection{Geo-indistinguishability for POI retrieval LBSs}
%
%Let us start by describing how geo-indistinguishability can be
%used to specify a subtle notion of privacy for LBS applications. 
%For that purpose, we first delineate the architecture of LBS applications that we consider in this work.  
We assume a simple
client-server architecture where users communicate via a trusted
mobile application (the client -- typically installed in a
smart-phone) with an unknown/untrusted LBS provider (the server --
typically running on the cloud). Hence, in contrast to other solutions
proposed in the literature, our approach does not rely
on trusted third-party servers.

\begin{figure}[t]
	\begin{center}
	{\small
	\begin{tabular}{lr}
	%\hline\vspace{-0.3cm}\\
	%\textbf{Sanitizing Algorithm for a Location} -- $\planar\\
	\hline \vspace{-0.3cm}\\ 
	\textbf{Input:} $x$ 									& \codeComment{point to sanitize}\\ 
	\hspace{0.8cm} $\epsilon$ 								& \codeComment{privacy parameter}\\
	\hspace{0.8cm} $u$, $v$, $\delta_\theta$, $\delta_r$    & \codeComment{precision parameters} \\
	\hspace{0.8cm} $\cala$									& \codeComment{acceptable locations}\\%
	\textbf{Output:} Sanitized version  $z$ of input $x$ \vspace{0.1cm}\\
	%
	%1. \ \  $q \gets \nicefrac{u}{\delta_\theta\diam(\cala)}$\\
	\multicolumn{2}{l}{1. \ \ $\epsilon'\gets$ max $\epsilon'$ satisfying Thm~\ref{theo:geo-dp} for $r_{\max}=\diam(\cala)$}\\
	2. \ \ draw $\theta$ unif. in $[0,2\pi)$ 							&\codeComment{draw angle}\\
	3. \ \ draw $p$ unif. in $[0,1)$, set $r \gets C_{\epsilon'}^{-1}(p)$  	&\codeComment{draw radius}\\
	4. \ \ $z\gets x + \langle r\cos(\theta), r\sin(\theta)\rangle$ 		&\codeComment{to cartesian, add vectors}\\
	5. \ \ $z \gets \operatorname{closest}(z,{\cal A})$						&\codeComment{truncation}\\
	6. \ \ \textbf{return} $z$\\
	\hline
	\end{tabular}
	}
	\end{center}
	%\smallskip
	\vspace{-6pt}
	\caption{The Planar Laplace mechanism $\planar$}
	\label{fig:algo_point}
\end{figure}

In the following we distinguish between \emph{mildly-location-sensitive}  and \emph{highly-location-sensitive} LBS applications. The former category corresponds to LBS applications offering a service that does not heavily rely on the precision of the
location information provided by the user. Examples of such applications are weather forecast applications 
%(forecast
%information for an approximate location is typically as good as
%forecast information for an exact location), 
%location-aware advertising/offers 
%(e.g., shops offering discounts typically care about
%users being nearby -- rather than their
%exact location), 
and  LBS applications for retrieval of certain kind of  POI (like gas stations). 
Enhancing this kind of LBSs with \geoind{}  is relatively straightforward as  it
only requires to obfuscate the user's location using the Planar Laplace mechanism (Figure~\ref{fig:algo_point}).

Our running example lies within the second category: For the user sitting at Caf\'{e} Les Deux Magots, information about restaurants nearby Champs {\'E}lys{\'e}es is considerably less valuable than information about restaurants around his location. 
Enhancing highly-location-sensitive LBSs  with privacy guarantees is more challenging.
Our approach consists on implementing the following three steps:

%Note that some LBS application may fall in one category or the
%other depending on the user's situation (e.g., if the user is by
%car or retrieving info for a very large area)

%Here the scenario is significantly more complicated than in the previous.

 \begin{figure}[tb]
      \centering
      \includegraphics[width=0.8\columnwidth]{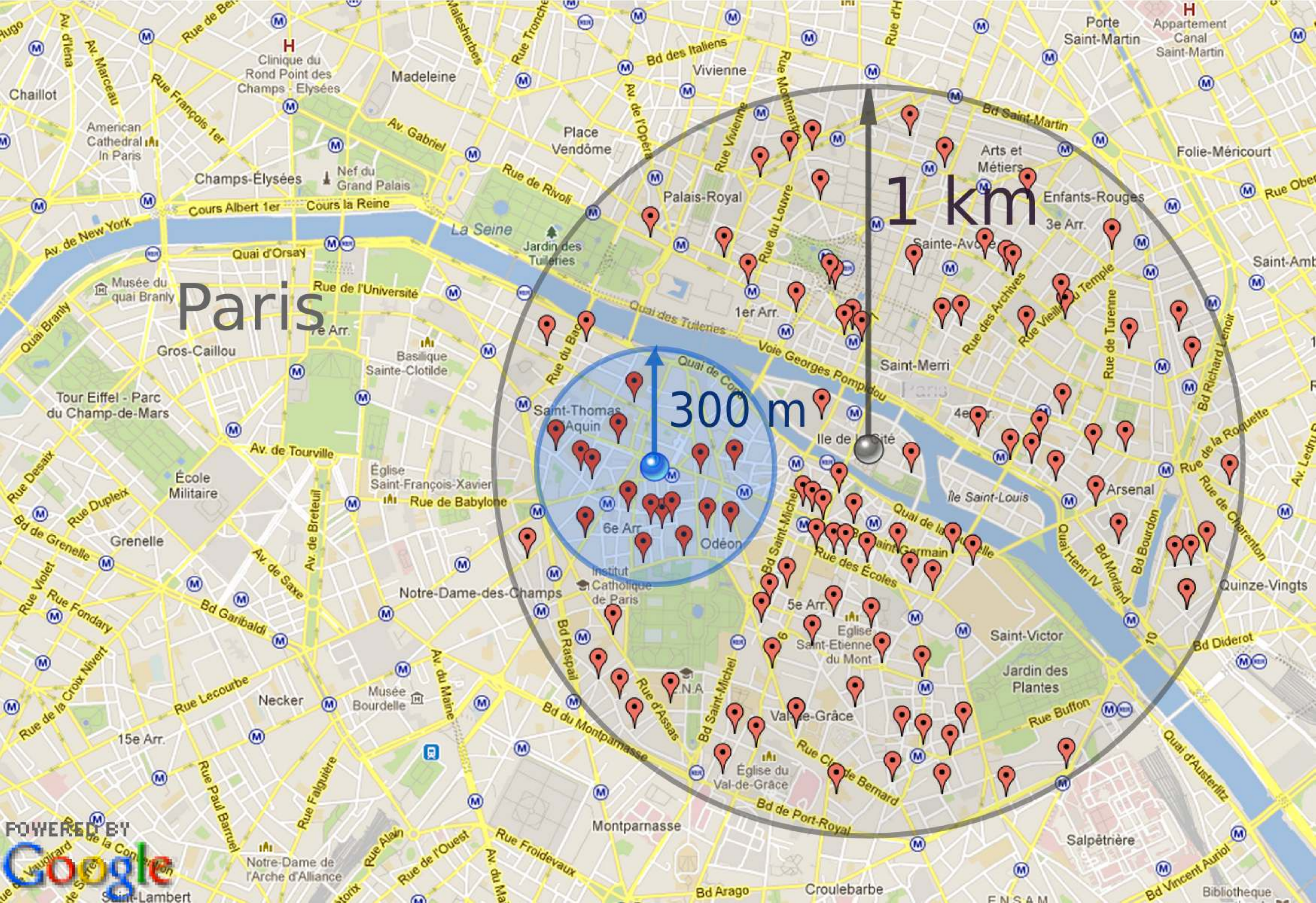}\\[1ex]
   \caption{AOI and AOR of $\mathbf{300}$ m and $\mathbf{1}$ km radius respectively.}
   \label{fig:privateLBS}
 \vspace{-10pt}
 \end{figure}

\begin{enumerate}
\item Implement the Planar Laplace mechanism (Figure~\ref{fig:algo_point})
on the client application in order to report to the LBS server the user's
obfuscated location $z$ rather than his real location $x$.
\item Due to the fact that the information retrieved from the server
is about POI nearby $z$, the area of POI information retrieval
should be increased. In this way, if the user wishes to obtain
information about POI within, say, $300$ m of $x$, 
%(we call this area the \emph{area of interest} of the user)
the client application should request information about POI
within, say, $1$ km of $z$.  Figure
\ref{fig:privateLBS}  illustrates this situation. 
%The user's current
%location $x$ is at caf\'e Les Deux Magots and the reported approximate
%location $z$ submitted by the client application is at about $500$ meters from $x$. 
We will refer to the blue circle  as \emph{area of interest} 
(AOI) and to the grey circle  as \emph{area of
retrieval} (AOR). 
 %(we call this area the \emph{retrieval area}).
%
\item Finally, the client application should filter the retrieved POI 
information (depicted by the pins within the area of retrieval in
Figure \ref{fig:privateLBS}) in order to provide to the user with
the desired information (depicted by pins within the user's area
of interest in Figure \ref{fig:privateLBS}).
\end{enumerate}
%
%The resulting client-server interaction is shown in Fig~\ref{fig:LBSarchitecturewithprivacy}.
%
%
% \begin{figure}[h]
%      \centering
%      \includegraphics[width=6cm]{figures/LBS_architecture_with_privacy.pdf}
%   \caption{LBS architecture}\label{fig:LBSarchitecturewithprivacy}
%   \vspace{-8pt}
% \end{figure}

Ideally, the AOI should always be fully contained in the AOR. 
Unfortunately, due to the probabilistic nature of our perturbation mechanism, this condition cannot be guaranteed (note that the AOR is centered on a randomly generated
location that can be arbitrarily distant from the real location). It is also worth noting that the client application cannot dynamically adjust the radius of the AOR in order to ensure
that it always contains the AOI as this approach would
completely jeopardize the privacy guarantees: on the one hand, the
size of the AOR would leak information about the
user's real location and, on the other hand, the LBS provider
would know with certainty that the user is located within the
retrieval area. Thus, in order to provide \geoind{}, the AOR
has to be defined \emph{independently} from the randomly generated
location.

\looseness -1
Since we cannot guarantee that the AOI is fully contained in the AOR, 
we introduce the notion of \emph{accuracy}, 
which measures the probability of such event.
In the following, we will refer 
to an LBS application in abstract terms, as characterized by a 
location perturbation mechanism $K$ and a fixed AOR radius.
We use $\mathit{rad}_R$ and $\mathit{rad}_I$ to denote the radius of the AOR and the AOI, respectively, 
and  $\calb(x, r)$ to denote the circle with center $x$ and radius $r$.

\subsection{On the accuracy of LBS{\sc s}}

Intuitively,  an LBS application is $(c, rad_I)$-accurate if the probability of the AOI to be fully contained in the AOR is bounded from below by a \emph{confidence factor} $c$.  Formally:

\begin{definition} [LBS application accuracy] An LBS application $(K, \mathit{rad}_R)$ is $(c, \mathit{rad}_I)$-accurate iff for all locations $x$ we have that $\calb(x, \mathit{rad}_I)$ is fully contained in $\calb(K(x), \mathit{rad}_R)$ with probability at least $c$.
\end{definition}

Given a privacy parameter $\epsilon$ and accuracy parameters ($c, \mathit{rad}_I$), our goal is to obtain an LBS application ($K, \mathit{rad}_R$) satisfying both \egeoind{} and ($c, \mathit{rad}_I$)-accuracy. As a perturbation mechanism, we use the Planar Laplace $\planar$ (Figure~\ref{fig:algo_point}), which satisfies \egeoind{}. As for $\mathit{rad}_R$, we aim at finding the minimum value validating the accuracy condition. Finding such minimum value is crucial to minimize the bandwidth overhead inherent to our proposal.
In the following we will investigate how to achieve this goal by \emph{statically} defining $\mathit{rad}_R$ as a function of the mechanism and the accuracy parameters $c$ and $\mathit{rad}_I$.

For our purpose, it will be convenient to use the  notion of  ($\alpha$, $\delta$)-usefulness, which was introduced in \cite{Blum:08:STOC}. 
A location perturbation mechanism $K$ is ($\alpha$, $\delta$)-\emph{useful} if for every
location $x$ the reported location 
$z=K(x)$ satisfies $d(x,z) \leq \alpha$ with probability at least $\delta$.
In the case of the Planar Laplace, it is easy to see that, by definition,  the $\alpha$ and $\delta$ values which 
express its usefulness are related by 
$C_\epsilon$ \footnote{%
	For simplicity we assume that $\epsilon'=\epsilon$ (see Figure~\ref{fig:algo_point}), since
	their difference is negligible under double precision.
},
the cdf of the Gamma distribution: 
%In other words, a ($\alpha$,$\delta$)-useful mechanism generates approximate locations $z$ within 
%distance $\alpha$ of the exact location $x$ with probability at least $1-\delta$.

\begin{observation}\label{obs:usefulness}
	For any $\alpha > 0$, $\planar$ is ($\alpha$, $\delta$)-useful if $\alpha\leq C^{-1}_\epsilon(\delta)$.
\end{observation}

Figure~\ref{fig:usefulness} illustrates the
($\alpha$, $\delta$)-usefulness of $\planar$
for $r\!=\!0.2$ (as in our running example) and various values of $\ell$ (recall that $\ell=\epsilon \,r$). 
It follows from the figure that a mechanism providing the privacy guarantees specified in our running example 
(\egeoind{}, with $\ell\!=\!\ln(4)$ and $r\!=\!0.2$) generates an
approximate location $z$ falling within $1$ km of the user's
location $x$ with probability $0.992$, falling within $690$ meters with probability $0.95$, falling within $560$
meters with probability $0.9$, and falling within $390$ meters
with probability $0.75$.

 \begin{figure}[tb]
      \centering
      \includegraphics[width=0.9\columnwidth]{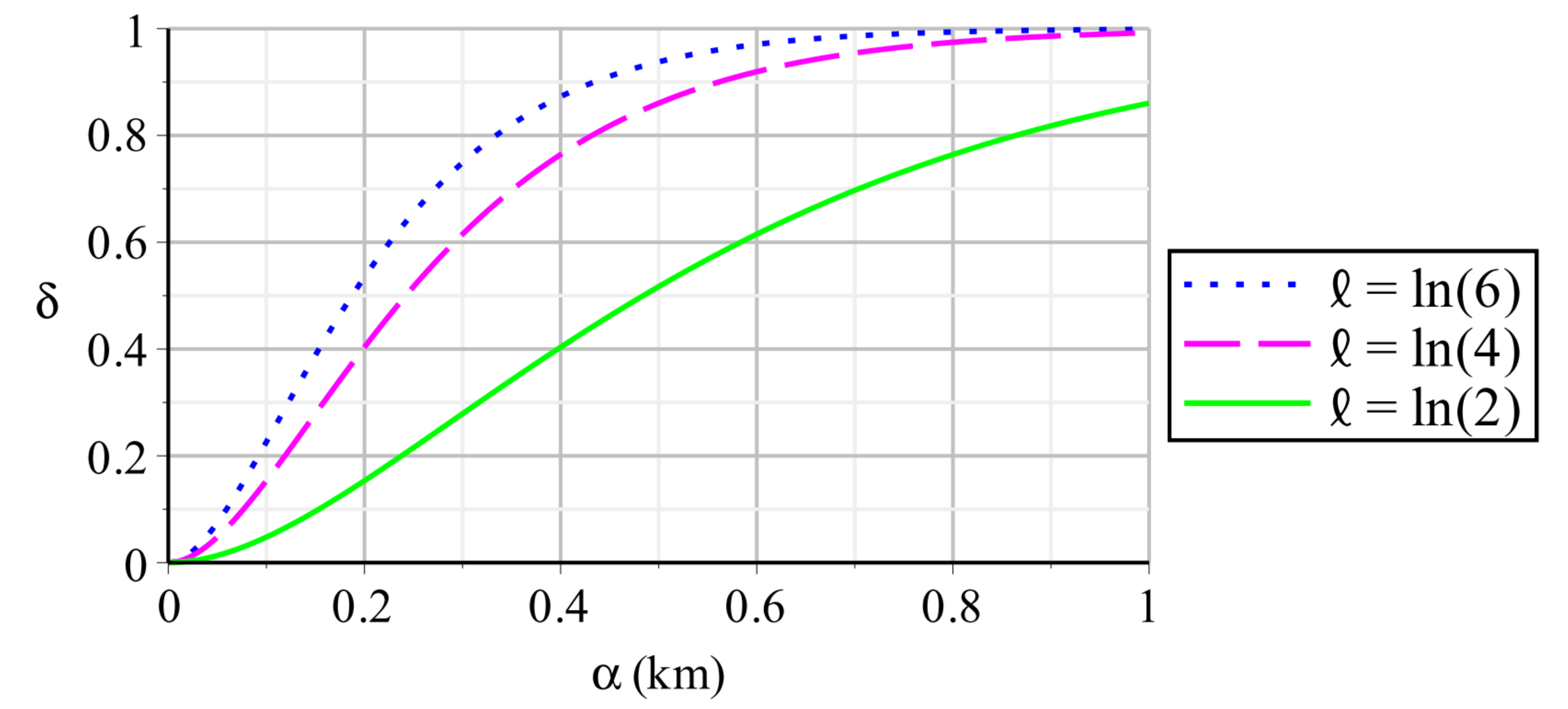}
   \caption{$(\alpha,\delta)$-usefulness for $r=0.2$ and various values of $\ell$.}
   \label{fig:usefulness}
   %\vspace{-10pt}
 \end{figure}

%it is desired to provide privacy while
%releasing information as accurate as possible (e.g., in standard
%differential privacy, one wishes to released the result of a query
%in a way as accurate as possible, while, of course, preserving the
%privacy of the users whose data is involved on the query). The
%same goal holds for the mechanism defined in Section
%\ref{sec:mechanism}, i.e., we wish to reveal a location as closed
%as possible to the user, while still providing
%geo-indistinguishability.
%
%We study the accuracy of our mechanism by means of a definition
%from the literature that can be used to measure the accuracy of a
%data perturbation mechanism (adapted to our location setting):
%$(\alpha,\delta)$-usefulness \cite{Blum:08:STOC}.
%
%
%\emph{A location obfuscation mechanism {\cal{K}} is $(\alpha,
%\delta)$-useful if for every location $x$, with probability at
%least $1-\delta$, the reported location $z = {\cal{K}}(x)$
%satisfies}
%\[
%d(x,z) \leq \alpha.
%\]
%
%Figure \ref{fig:usefulness} illustrates how the mechanism defined
%in Section \ref{sec:mechanism} behaves with respect to this notion
%of ``mechanism accuracy''.

We now have all the necessary ingredients to determine the desired $\mathit{rad}_R$:
By definition of usefulness, if $\planar$ is ($\alpha$, $\delta$)-useful then the LBS application $(\planar,\mathit{rad}_R)$ is  $(\delta, \mathit{rad}_I)$-accurate 
if $\alpha \leq \mathit{rad}_R - \mathit{rad}_I$. The converse also holds if $\delta$ is maximal. 
By Observation~\ref{obs:usefulness}, we then have:
\begin{proposition} 
	The LBS application $(\planar, \mathit{rad}_R)$ is
	$(c$, $\mathit{rad}_I)$-accurate if $\mathit{rad}_R\geq \mathit{rad}_I + C_\epsilon^{-1}(c)$.
	\label{prop:aor}
\end{proposition}
Therefore, it is sufficient to set  $\mathit{rad}_R=\mathit{rad}_I + C_\epsilon^{-1}(c)$. 

Coming back to our running example ($\epsilon= \ln(4)/ 0.2$ and $\mathit{rad}_I=0.3$), taking a confidence factor $c$ of, say, $0.95$, 
leads to a $(0.69, 0.95)$-useful mechanism (because $C_\epsilon^{-1}(c)=0.69$). Thus, $(\planar, 0.99)$ is both $\nicefrac{\ln(4)}{0.2}$-geo-indistinguishable and ($0.95, 0.3$)-accurate.

\subsection{Bandwidth overhead analysis}

As expressed by Proposition~\ref{prop:aor}, in order to implement an LBS application enhanced with \geoind{} and accuracy it suffices to use the Planar Laplace mechanism and retrieve POIs for an enlarged radius $\mathit{rad}_R$. For each query made from a location $x$, the application needs to (i) obtain $z=\planar(x)$, (ii) retrieve POIs for AOR $=\calb(z, \mathit{rad}_R)$, and (iii) filter the results from AOR to AOI (as explained in step 3 above). Such implementation is straightforward and computationally efficient for modern smart-phone devices. In addition, it provides great flexibility to application developer and/or users to specify their desired/allowed level of privacy and accuracy. This, however, comes at a cost: bandwidth overhead.

In the following we turn our attention to investigating the bandwidth overhead yielded by our approach. We will do so in two steps: first we investigate how the AOR size increases for different privacy and LBS-specific parameters, and then we investigate how such increase translates into bandwidth overhead.

Figure~\ref{fig:usefulness2} depicts the overhead of the AOR versus the AOI (represented as their ratio) when varying the level of confidence ($c$) and privacy ($\ell$) and for fixed values $\mathit{rad}_I=0.3$ and $r=0.2$. The overhead increases slowly for levels of confidence up to $0.95$ (regardless of the level of privacy) and increases sharply thereafter, yielding to a worst case scenario of a about $50$ times increase for the combination of highest privacy ($\ell=\log(2)$) and highest confidence ($c=0.99$). 

 \begin{figure}[tb]
      \centering
      \includegraphics[width=0.9\columnwidth]{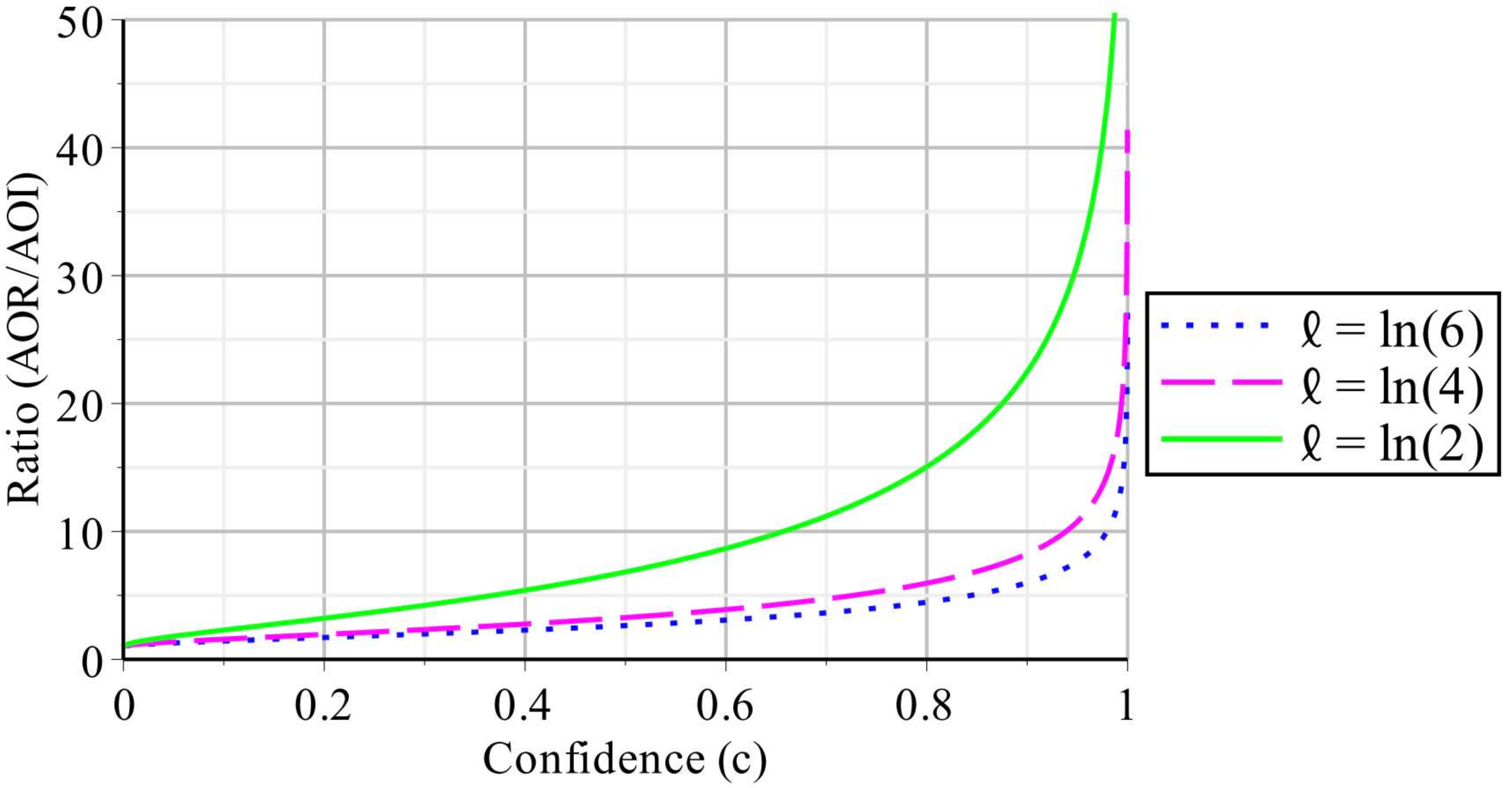}
   \caption{AOR vs AOI ratio for various levels of privacy and accuracy (using fixed $r=0.2$ and $rad_I=0.3$).}
   \label{fig:usefulness2}
 \end{figure}

In order to understand how the AOR increase translates into bandwidth overhead, we now investigate the density (in km$^2$) and size (in KB) of POIs by means of the Google Places API \cite{GooglePlacesApi}. This API allows to retrieve POIs' information for a specific \texttt{location}, \texttt{radius} around the location, and POI's \texttt{type} (among many other optional parameters). For instance, the HTTPS request:

	\begin{small}
	\vspace{-10pt}
	\begin{verbatim}
		https://maps.googleapis.com/maps/api/place/nearby
			search/json?location=48.85412,2.33316 &
			radius=300 & types=restaurant & key=myKey
	\end{verbatim}
	\end{small}
returns information (in JSON format) including location, address, name, rating, and opening times for all restaurants up to $300$ meters from the location ($48.85412$,\ $2.33316$) -- which corresponds to the coordinates of Caf\'{e} Les Deux Magots in Paris.

We have used the APIs \texttt{nearbysearch} and \texttt{radarsearch} to calculate the average number of POIs per km$^2$ and the average size of POIs' information (in KB) respectively. We have considered two %three 
queries: restaurants in Paris, and restaurants in Buenos Aires.
%and cinemas in Paris. 
Our results show that there is an average of $137$ restaurants per km$^2$ in Paris and $22$ in Buenos Aires, while %about $2$ cinemas per km$^2$ in Paris. Furthermore, 
the average size per POI is $0.84$ KB. 

Combining this information with the AOR overhead depicted in Figure~\ref{fig:usefulness2}, we can derive the average bandwidth overhead for each query and various combinations of privacy and accuracy levels. For example, using the parameter combination of our running example (privacy level $\epsilon = \log(4)/0.2$, and accuracy level $c=0.95, \mathit{rad}_I=0.3$) we have a $10.7$ ratio for an average of $38$ ($\backsimeq (137 / 1000^2) \times (300^2 \times \pi)$) restaurants in the AOI. Thus the estimated bandwidth overhead is $39 \times (10.7-1) \times 0.84$KB $\backsimeq 318$~KB.

Table \ref{tab:bandwidth} shows the bandwidth overhead for restaurants in Paris and Buenos Aires for the various combinations of privacy and accuracy levels. Looking at the worst case scenario, from a bandwidth overhead perspective, our combination of highest levels of privacy and accuracy (taking $\ell=\log(2)$ and $c=0.99$) with the query for restaurants in Paris (which yields to a large number of POIs -- significantly larger than average) results in a significant bandwidth overhead (up to $1.7$MB). Such overhead reduces sharply when decreasing the level of privacy (e.g., from $1.7$ MB to $557$ KB when using $\ell=\log(4)$ instead of $\ell=\log(2)$). For more standard queries yielding a lower number of POIs, in contrast, even the combination of highest privacy and accuracy levels results in a relatively insignificant bandwidth overhead.

\begin{table}
\begin{center}
    \begin{tabular}{| c  c | c | c | c |}
    \hline
    \multicolumn{2}{|c|}{{\multirow{2}{*} {Restaurants}}} & \multicolumn{3}{ c|} {\textbf{Accuracy}}\\
    \multicolumn{2}{|c|}{{\multirow{2}{*} {in Paris}}}& \multicolumn{3}{ c|} {$\mathit{rad}_I=0.3$}\\ 
    & & \multicolumn{1}{c}{$c=0.9$} & \multicolumn{1}{c}{$c=0.95$} & \multicolumn{1}{c|}{$c=0.99$} \\  \hline
     \multicolumn{1}{|l}{\multirow{2}{*}{\textbf{Privacy}} }  & $\ell\!=\!\log(6)$ & $162$ KB & $216$ KB & $359$ KB \\
    \multicolumn{1}{|l}{\multirow{2}{*}{$r\!=\!0.2$} } & $\ell\!=\!\log(4)$ & $235$ KB &  $318$ KB & $539$ KB\\
    & $\ell\!=\!\log(2)$ & $698$ KB & $974$ KB & $1.7$ MB \\
    \hline
    \end{tabular}

\medskip

    \begin{tabular}{| c  c | c | c | c |}
    \hline
    \multicolumn{2}{|c|}{{\multirow{2}{*} {Restaurants}}} & \multicolumn{3}{ c|} {\textbf{Accuracy}}\\
    \multicolumn{2}{|c|}{{\multirow{2}{*} {in Buenos Aires}}}& \multicolumn{3}{ c|} {$\mathit{rad}_I=0.3$}\\ 
    & & \multicolumn{1}{c}{$c=0.9$} & \multicolumn{1}{c}{$c=0.95$} & \multicolumn{1}{c|}{$c=0.99$} \\  \hline
     \multicolumn{1}{|l}{\multirow{2}{*}{\textbf{Privacy}} }  & $\ell\!=\!\log(6)$ & $26$ KB & $34$ KB & $54$ KB \\
    \multicolumn{1}{|l}{\multirow{2}{*}{$r\!=\!0.2$} } & $\ell\!=\!\log(4)$ & $38$ KB &  $51$ KB & $86$ KB\\
    & $\ell\!=\!\log(2)$ & $112$ KB & $156$ KB & $279$ KB \\
    \hline
    \end{tabular}
 \caption{Bandwidth overhead for restaurants in Paris and in Buenos Aires for various levels of privacy and accuracy.}
\label{tab:bandwidth}
\end{center}
\vspace{-10pt}
\end{table}

Concluding our bandwidth overhead analysis, we believe that the overhead necessary to enhance an LBS application with \geoind{} guarantees is not prohibitive even for scenarios resulting in high bandwidth overhead (i.e., when combining very high privacy and accuracy levels with queries yielding a large number of POIs). Note that $1.7$MB is comparable to $35$ seconds of Youtube streaming or $80$ seconds of standard Facebook usage \cite{VodafoneMobileDataUsageStats}. Nevertheless, for cases in which minimizing bandwidth consumption is paramount, we believe that trading bandwidth consumption for privacy (e.g., using $\ell=\log(4)$ or even $\ell=\log(6)$) is an acceptable solution.

\subsection{Further challenges: using an LBS multiple times}

%After describing how to provide geo-indistinguishability guarantees to users querying an LBS application a \emph{single} time, we now discuss how to extend our solution to the case in which users wish to perform \emph{multiple} queries. %In practice however, it is common for users to perform multiple queries (typically from multiple different locations) %or, in a similar spirit, some LBSs applications require the user to share/expose their locations in a \emph{continuous} fashion, i.e.., submit information several times in a short period of time. 
%
%In this scenario, the mechanism should protect multiple locations rather than one. But, what does it mean to enjoy privacy for multiple locations? 
As discussed in Section \ref{sec:multiple-locations}, \geoind{} can be naturally extended to multiple locations. In short, the idea of being \emph{$\ell$-private within $r$} remains the same but for all locations simultaneously. In this way the locations, say, $x_1$, $x_2$ of a user employing the LBS twice remain indistinguishable from all pair of locations at (point-wise) distance at most $r$ (i.e., from all pairs $x_1'$, $x_2'$ such that $d(x_1,x'_1)\leq r$ and $d(x_2,x'_2)\leq r$). 

A simple way of obtaining \geoind{} guarantees when performing multiple queries is to employ our technique for protecting single locations to \emph{independently} generate approximate locations for each of the user's locations. In this way, a user performing $n$ queries via a mechanism providing \egeoind{} enjoys \egeoind[n \epsilon] (see Section~\ref{sec:multiple-locations}).

This solution might be satisfactory when the number of queries to perform
remains fairly low, but in other cases impractical, due to the privacy degradation.
It is worth noting that the canonical technique for achieving standard
differential privacy (based on adding noise according to the Laplace
distribution) suffers of the same privacy degradation problem ($\epsilon$
increases linearly on the number of queries). Several articles in the
literature focus on this problem (see \cite{Roth:10:STOC} for instance). We
believe that the principles and techniques used to deal with this problem for
standard differential privacy could be adapted to our scenario (either directly
or motivationally). 
\section{Comparison with other methods}
\label{sec:comparison}

In this section we compare the  performance of our mechanism with that of  other ones proposed in the literature. 
Of course it is not interesting to make a comparison in terms of \geoind, since  other mechanisms usually do not satisfy this property. 
We consider, instead, the  (rather natural) Bayesian notion  of privacy proposed in  \cite{Shokri:12:CCS},  and the trade-off with respect to the \emph{quality of service} measured according to \cite{Shokri:12:CCS}, and also with respect to  the notion of accuracy illustrated in the previous section. 
%A proper way to show the suitability of our mechanism is to perform a meaningful comparison with other methods existing in the literature. One of the most relevant works in this context is the one of Shokri et al. (\cite{Shokri:12:CCS}), where the authors develop a mechanism that maximizes the expected estimation error of an adversary implementing an optimal inference attack, while still achieving the minimum quality of service required by the user. Another common approach is to divide the map into regions of adequate size, and report the center of the region in which the user is located.
%In this section we present the results of our experiments,
%where we compare the performance of our approach against these two mechanisms.
 
 The mechanisms  that we compare with ours are: 
\begin{enumerate}
\item
The obfuscation mechanism presented in  \cite{Shokri:12:CCS}. This mechanism works on discrete locations,    
called  \emph{regions}, and, 
like ours, it reports a location (region) selected randomly according to a probability distribution that  
depends on the user's location. The distributions are generated automatically by a tool which is
designed to provide optimal privacy for a given quality of service and a given adversary 
(i.e., a given prior, representing the side knowledge of the adversary). 
%%\subsection{The optimal mechanism}
%
%In \cite{Shokri:12:CCS}, an obfuscation mechanism over regions is presented. This mechanism has the property of being \emph{optimal} for a given user and required utility, meaning that it maximizes the location privacy for that user while still providing a minimum required quality of service.
%
%The mechanism is constructed by solving a linear program that takes into account the \emph{user profile} (that is, the probability distribution of the user over the regions), the metrics used to calculate the privacy and the quality degradation, and the maximum tolerable utility loss. It is worth noting that the mechanism obtained is then user-specific, and depends on knowing a good estimation of the user profile. 
It is important to note that in presence of a different adversary the optimality is not guaranteed. 
This dependency on the prior  is a key difference with respect to our approach, which abstracts from the adversary's side information. 
\item
A simple cloaking mechanism.
In this approach, the area of interest is assumed to be partitioned in \emph{zones}, 
whose size depends on the level of privacy we want yo achieve. 
The mechanism then reports the  {zone} in which the exact location is situated. 
This method satisfies $k$-anonymity where $k$ is the number of locations within each zone.
\end{enumerate}

In both cases we  need to divide the area of interest  into a finite number of regions, representing the possible locations. 
We  consider for simplicity a grid, and, more precisely, a $9 \times 9$ grid consisting of $81$ square regions of $100$ m of side length.
In addition, for the cloaking method,  we overlay a grid of $3 \times 3 = 9$ zones. 
Figure~\ref{fig:regions-adv1} illustrates the setting: the regions are the small squares with black borders. 
In the cloaking  method, the zones are the larger squares with blue borders. 
For instance,  any point situated in one of the regions $1, 2, 3, 10, 11, 12, 19, 20$ or $21$, would be reported as zone $1$. 
We assume that each zone is represented by the central region.  Hence, in the above example, the reported region would be~$11$.

\paragraph{Privacy and Quality of Service}
As already stated, we will use the metrics for  privacy and for the quality of service  proposed in \cite{Shokri:12:CCS}.

The first metric is called \emph{Location Privacy} ($\mathit{LP}$) in  \cite{Shokri:12:CCS}. The idea is to measure it in terms of 
the  \emph{expected estimation error} of a ``rational'' Bayesian adversary.  The adversary is assumed to have some side knowledge, 
expressed in terms of a probability distribution on the regions, which represents the \emph{a priori} probability that the user's location is situated in that region. 
The adversary   tries to make the best use of such prior information, and   combines it with the information provided by the mechanism (the reported region), so to guess a location (remapped region) which is as close as possible to the one where the user really is. More precisely, the goal is to infer a region that, in average, minimizes the distance from the user's exact location.

Formally, $\mathit{LP}$  is defined as:
 \[
\mathit{LP} = \sum_{r, r', \hat{r} \in R} \pi (r) \K(r)(r') h(\hat{r}|r') d(\hat{r}, r)
\]
where $R$ is the set of all regions, $\pi$ is the prior distribution over the regions, $\K(r)(r')$ gives the probability that the real region $r$ is reported by the mechanism as $r'$, $h(\hat{r}|r')$ 
represents the probability that the reported region $r'$ is remapped into $\hat{r}$, 
in the optimal remapping $h$, and $d$ is the distance between regions.
``Optimal'' here means that $h$ is chosen so to minimize the above expression, which, we recall, represents the expected distance between the user's exact location and the location guessed by the adversary.

As for the quality of service, the idea  in \cite{Shokri:12:CCS} is to quantify its opposite, the \emph{Service Quality Loss} ({\it SQL}), in terms of 
the expected distance  between the reported location and the user's exact location. In other words, 
the service provider is supposed to offer a quality proportional to the accuracy of the location that he receives. 
Unlike the adversary, he is not expected to  have any prior knowledge and he is not expected to guess a  location different from the reported one. 
Formally:
 \[
\mathit{SQL} = \sum_{r, r' \in R} \pi (r) \K(r)(r')  d(r', r)
\]
where $\pi$, $\K(r)(r')$ and $d$ are as above.

It is worth noting that for the optimal mechanism in \cite{Shokri:12:CCS} \emph{SQL} and LP coincide (when the mechanism is used in presence of the same adversary for which it has been designed), i.e. the adversary does not need to make any remapping.

%\footnote{The Location Privacy Meter (\cite{Shokri:11:SP}) was also taken into account as a possible option to perform the comparison, but in this tool the location privacy is determined with respect to an adversary which do not necessarily make an optimal remap.}.

\begin{figure}[tb]
      \centering
      \includegraphics[width=0.4\columnwidth]{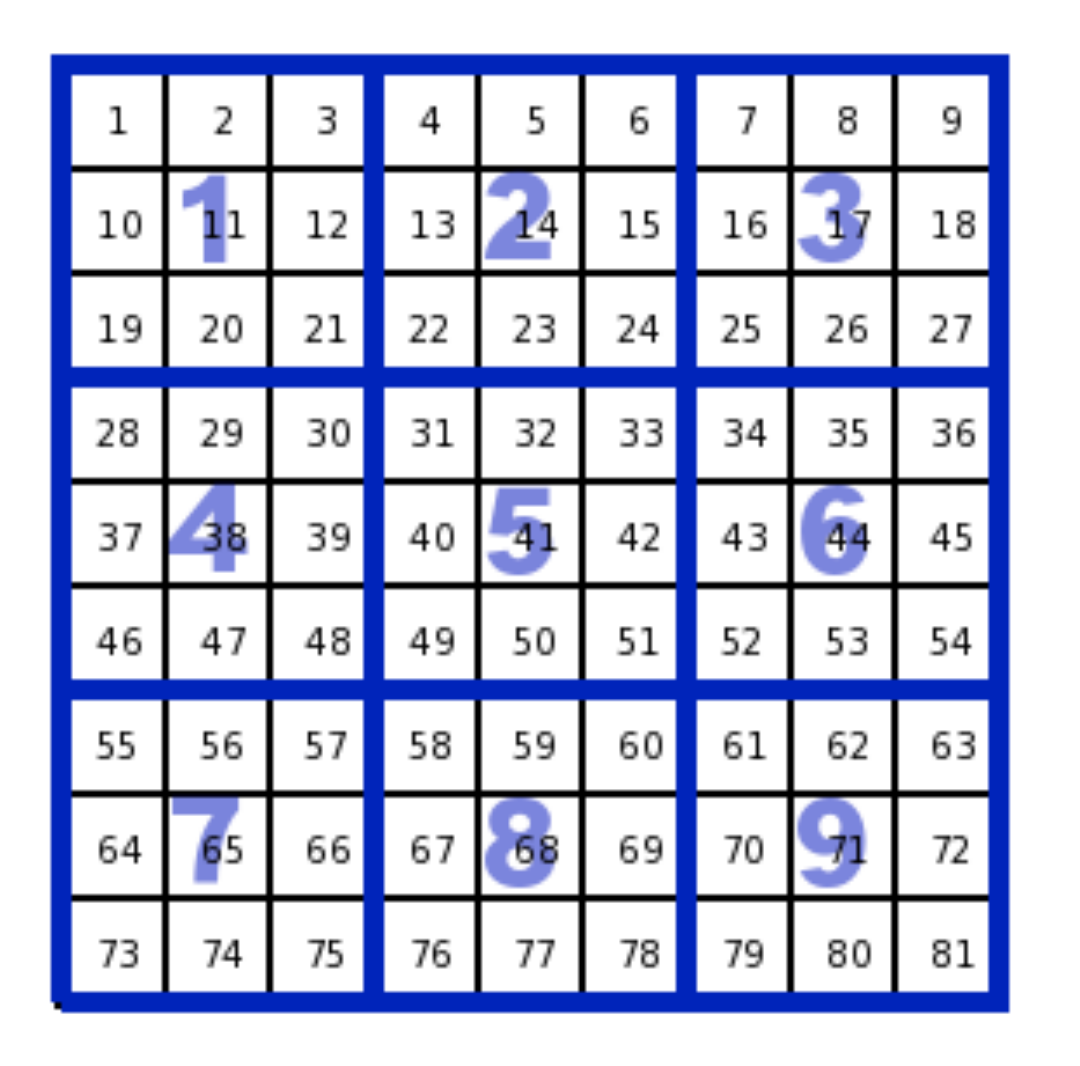}
   \caption{The division of the map into regions and zones.}
   \label{fig:regions-adv1}
  % \vspace{-10pt}
 \end{figure}

 \paragraph{Comparing the $\mathit{LP}$ for a given $\mathit{SQL}$}
%Our mechanism is already tuned to be mapped into a discrete grid (see Section 4). Let us make more precise the way we proceed in order to compute \emph{LP}: 
%\begin{enumerate}
%\item When obfuscating a region $r$, we considered the location to be the center of that region.
%\item We then apply our usual mechanism to this location, and get an obfuscated location as result.
%\item Finally, we remap this obfuscated location to the closest region $r'$ of the grid, and report $r'$ as the obfuscated region.
%\end{enumerate}
In order to compare the three mechanisms, we set the parameters of each mechanism  in such a way that the  \emph{SQL} is the same for all of them, and we compare their \emph{LP}. 
As already noted, for the optimal mechanism in \cite{Shokri:12:CCS} \emph{SQL} and \emph{LP} coincide, i.e. the optimal remapping is the identity, when the mechanism is used in presence of the same adversary for which it has been designed. 
It turns out that, when the adversary's prior is the uniform one,  \emph{SQL} and \emph{LP} coincide also for our mechanism and for the cloaking one.

 We note that for the cloaking mechanism the \emph{SQL} is fixed and it is $107.03$ m. In our experiments we   fix the  value of \emph{SQL} to be that one for all the mechanisms. 
We find that in order to obtain such \emph{SQL} for our mechanism we need to set  $\epsilon = 0.0162$ (the difference with $\epsilon'$ in this case is negligible). The mechanism of \cite{Shokri:12:CCS} is generated by using the tool explained in the same paper.

 Figure~\ref{fig:priors} illustrates the priors that we consider here: in each case, the probability distribution is accumulated in the regions in the purple area, and distributed uniformly over them. 
Note that it is not interesting to consider 
the uniform distribution over  the whole map, since, as explained before,  on that prior all the mechanisms under consideration give the same result.

 \begin{figure}[tb]
      \centering
      \includegraphics[width=\columnwidth]{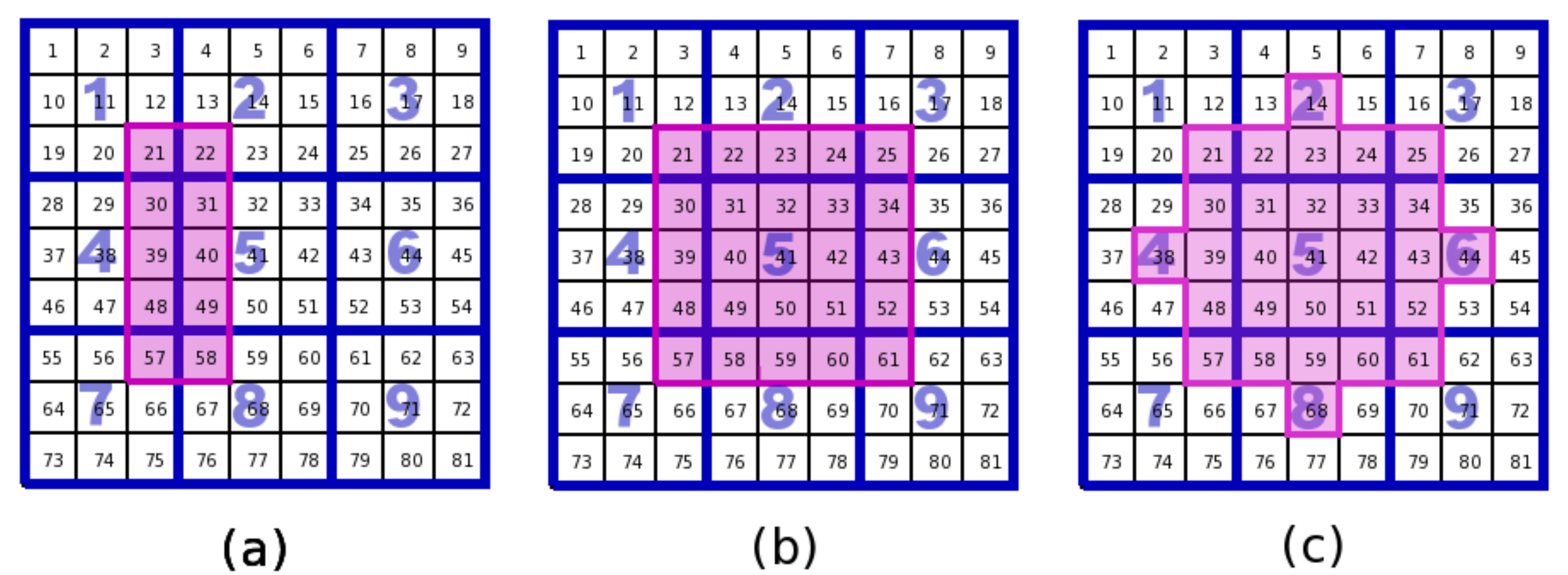}
   \caption{Priors considered for the experiments.}
   \label{fig:priors}
  % \vspace{-10pt}
 \end{figure}
 
 Figure~\ref{fig:comp1} illustrates the results we obtain in terms of \emph{LP}, where (a), (b) and (c)  correspond to the priors in Figure~\ref{fig:priors}. The optimal mechanism is considered in two instances: 
 the one designed exactly for the prior for which it is used (``optimal-rp'', where ``rp'' stands for real prior), and the one designed for the uniform distribution on all the map (``optimal-unif'', which is not necessarily optimal for the priors considered here). As we can see, the Planar Laplace mechanism offers the best \emph{LP} among the mechanisms which do not depend on the prior, or, as in the case of optimal-unif, are designed with a fixed prior. When the prior has a more circular symmetry the performance approaches the one of optimal-rp (the optimal mechanism).

 \begin{figure}[t]
      \centering
      \includegraphics[width=\columnwidth]{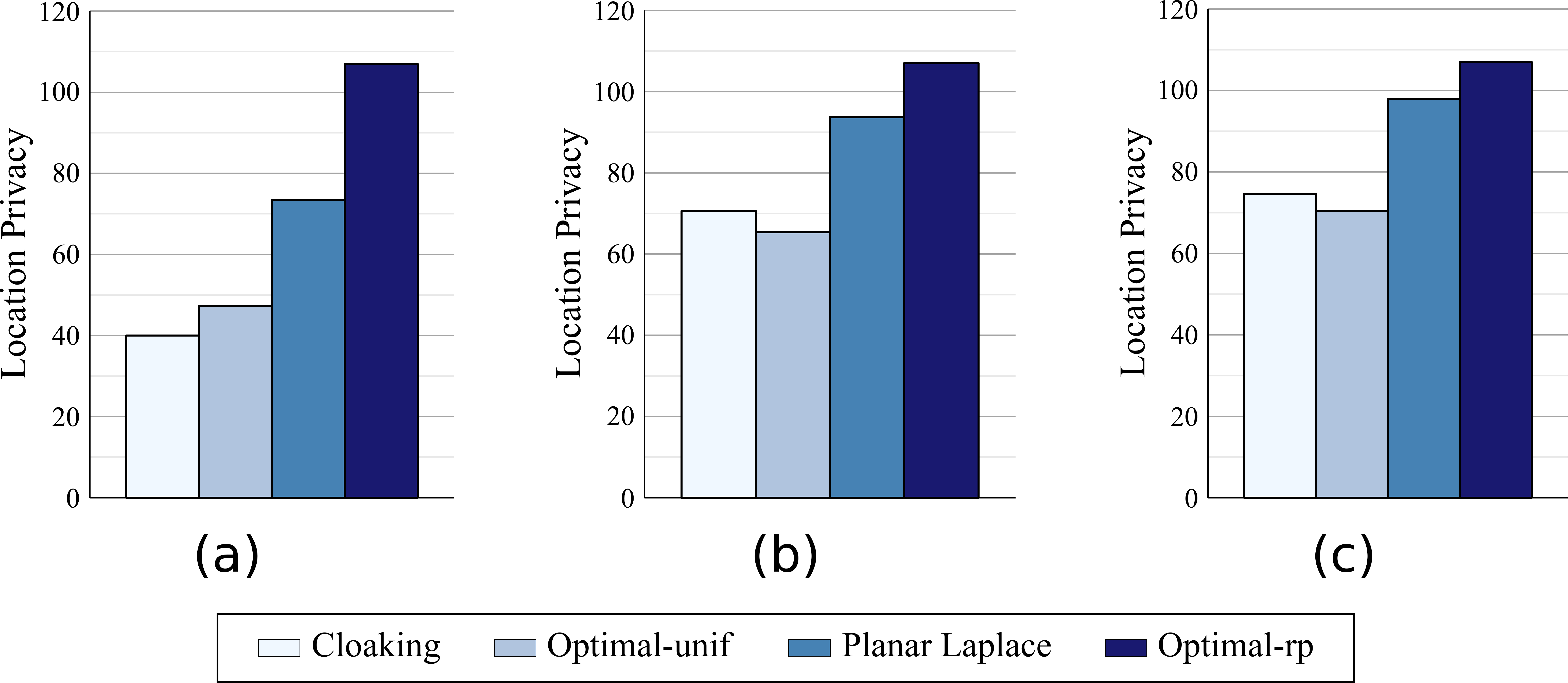}
   \caption{Location Privacy for $\boldsymbol{SQL=107.03}$ m.}
   \label{fig:comp1}
   \vspace{-10pt}
 \end{figure}

\paragraph{Comparing the $\mathit{LP}$ for a given accuracy}
The  \emph{SQL} metric defined above is a reasonable metric, but it does not cover all natural notions of quality of service. 
In particular, in the case of LBSs, an important criterion to take into account is the additional bandwidth usage.
%, i.e., the number of
%additional POIs that the user needs to retrieve in order to get the  ones he wants. 
%Such number is proportional to the size of area of retrieval (AOR) that we need to set in order to cover the area of interest (AOI) with  an acceptable probability. Intuitively, this size is proportional to the \emph{square of the distance} (rather than the simple distance) between the exact location and the reported one. 
Therefore, we make now a comparison using  the notion of accuracy, which, as explained in previous section, provides a good criterion to evaluate the performance in terms of bandwidth. Unfortunately we cannot compare our mechanism to the one of \cite{Shokri:12:CCS} under this criterion, 
because  the construction of the latter is tied to the \emph{SQL}.
Hence, we only compare our mechanism with the cloaking one. 
% Kostas: the accuracy constraint might indeed be non-linear, but C^-1 has nothing to do with it
%In fact, the tool is based on linear programming, and while fixing a bound for  the  \emph{SQL} constitutes a linear constraint, 
%fixing a bound for the accuracy, as discussed in previous section, involves  $C_\epsilon^{-1}$, which is non linear. 

We recall that an LBS application $(K,\mathit{rad}_R)$ is $(c,\mathit{rad}_I)$-accurate if for every
location $x$  the probability that the area of interest (AOI) is fully contained in the area of retrieval (AOR) is at least $c$.
We need to fix $\mathit{rad}_I$ (the radius of the AOI),  $\mathit{rad}_R$ (the radius of the AOR), and $c$ so that the condition of accuracy is satisfied for both methods, and then compute the respective \emph{LP} measures.
Let us fix  $\mathit{rad}_I=200$ m, and let us choose a large confidence factor, 
say, $c= 0.99$. 
As for $\mathit{rad}_R$, it will be determined by the cloaking method.

Since the cloaking mechanism is deterministic, in order for the condition to be satisfied
the  AOR for a given location $x$ must extend around the zone of $x$ 
by at least $\mathit{rad}_I$, 
In fact, $x$ could be in the border of the zone. Given that the 
cloaking method reports the center of the zone, and that the distance between the center 
and the border (which is equal to the distance between the center and any of the corners)
is $\sqrt{2} \cdot 150$ m, we derive that $\mathit{rad}_R$ must be at least 
$(200+\sqrt{2} \cdot 150)$ m. 
Note that in the case of this method the accuracy  is independent from the value of $c$.
It only depends on the difference between 
$\mathit{rad}_R$ and $\mathit{rad}_I$, which in turns depends 
on the length $s$ of the side of the region:
if the difference is at least $\sqrt{2} \cdot \nicefrac{s}{2}$,
then the condition is satisfied (for every possible $x$) with probability $1$. 
Otherwise, there will be some $x$ for which the condition is not satisfied 
(i.e., it is satisfied with probability $0$). 

In the case of our method, on the other hand, the accuracy condition depends on $c$ and on  $\epsilon$. More precisely, as we have seen 
in previous section, the condition is satisfied if and only if  $C_\epsilon^{-1}(c) \leq \mathit{rad}_R - \mathit{rad}_I$.
Therefore, for fixed $c$, the maximum $\epsilon$  only depends on the difference between 
$\mathit{rad}_R$ and $\mathit{rad}_I$ and is determined by the equation $C_\epsilon^{-1}(c) = \mathit{rad}_R - \mathit{rad}_I$.
For the above values of $\mathit{rad}_I$,  $\mathit{rad}_R$, and $c$, it turns out that $\epsilon=0.016$. 

We can now compare the  \emph{LP} of the two mechanisms with respect to the  three priors above.
Figure~\ref{fig:comp2} illustrates the results. As we can see, our mechanism outperforms the cloaking mechanism 
in all the three cases. 

For different values of $\mathit{rad}_I$ the situation does not change: as explained above,  the cloaking method always forces 
$\mathit{rad}_R$ to be larger than $\mathit{rad}_I$ by (at least)  $\sqrt{2} \cdot 150$ m, and 
 $\epsilon$ only depends  on this value. 
For smaller values of $c$, on the contrary, the situation changes, and becomes more favorable for our method. 
In fact, as argued above, the situation remains the same for the cloacking method (since its accuracy does not depend on $c$), 
while $\epsilon$ decreases (and consequently $\mathit{LP}$ increases) as $c$ decreases. In fact, for a fixed $r=\mathit{rad}_R-\mathit{rad}_I$, we have $\epsilon=C^{-1}_r(c)$.
This follows from  $r=C^{-1}_\epsilon(c)$ and from the fact that $r$ and $\epsilon$, in the expression that defines  $C_\epsilon(r)$, 
are interchangeable.

 \begin{figure}[tb]
      \centering
      \includegraphics[width=\columnwidth]{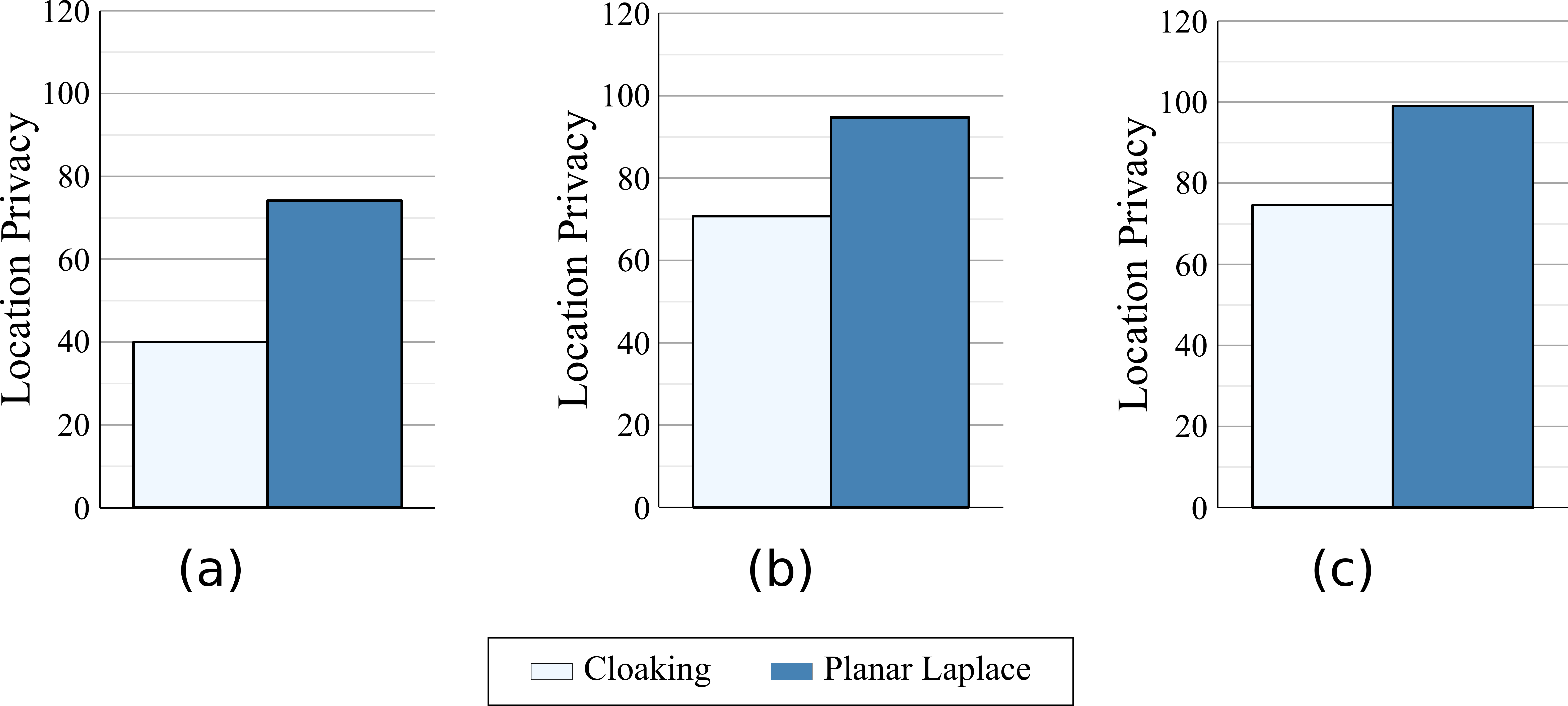}
   \caption{Location Privacy for  $\mathit{rad}_R=(\sqrt{2} \cdot 150+200)$ m and $c= 0.99$.}
   \label{fig:comp2}
  %\vspace{-10pt}
 \end{figure}

\section{Related Work}

Much of the related work has been already discussed in Section~\ref{sec:existing-notions},  here we 
only mention the works that were not reported there. 
%We refer to  \cite{Terrovitis:11:SIGKDD} for an excellent survey on privacy methods for geolocation. 
There are excellent works and surveys \cite{Terrovitis:11:SIGKDD,Krumm:09:PUC,Shin:12:WC} that summarize the different threats, methods, and guarantees
in the context of location privacy. 

LISA \cite{Chen:09:Thesis} provides location privacy by preventing an attacker
from relating any particular point of interest (POI) to the user's location.
That way, the attacker cannot infer which POI the user will visit next. The
privacy metric used in this work is $m$-\emph{unobservability}. The method
achieves $m$-unobservability if, with high probability, the attacker cannot
relate the estimated location to at least $m$ different POIs in the proximity.
%This method  does not take into
%account the attacker's prior knowledge, and it also requires to maintain a
%database of POIs in the mobile device in order to provide the required privacy.

SpaceTwist \cite{Yiu:08:ICDE} reports a fake location (called the ``anchor'')
and queries the geolocation system server incrementally for the nearest
neighbors of this fake location until the $k$-nearest neighbors of the real
location are obtained. 
%The shortcomings for this method are that the privacy
%guarantees do not consider the attacker a-priori information, and that the
%attacker can bound the real location to a region $\Omega$.

%There are also some works whose main goal is to provide accurate results for
%data mining algorithms while preserving location privacy of the user. Gidofalvi
%et al. \cite{Gidofalvi:07:MDM} use grid-based anonymization, although the
%privacy guarantees are mainly experimental. 

%Differential privacy using a metric other than the Hamming distance (which
%corresponds to standard differential privacy) appears in several works. In
%\cite{Reed:10:ICFP}, various metrics are used associated to different data types
%in a programming language, but the ultimate goal is to show that a program
%satisfies the standard notion. In \cite{Hardt10STOC}, distance between databases
%is measured by the $\ell_1$-metric, while \cite{Dwork:12:ITCS} uses a metric
%that captures the similarity between individuals, based on their membership in
%some group.

In a recent paper \cite{Mironov:12:CCS} it has been shown that, due to finite precision and rounding effects of floating-point operations,  
the standard implementations of the Laplacian mechanism result in an irregular distribution which causes
 the loss of the property of differential privacy. 
In  \cite{Gazeau:13:QAPL} the study has been extended to the planar Laplacian, 
and to any kind of finite-precision semantics. The same paper proposes a solutions for the truncated version of the planar laplacian,  
based on a snapping meccanism, which maintains the level of privacy at the cost of introducing an additional amount of noise. 

\section{Conclusion and future work}

In this paper we have  presented a framework for achieving 
privacy in location-based applications, taking into account the desired level of 
protection as well as the side-information that the attacker might have. 
%, trying to overcome all the inconveniences mentioned in the previous
%approaches. 
The core of our proposal is a new notion of privacy, that we call  geo-indistinguishability, 
and a method, based on a  bivariate version of the Laplace function, 
to perturbate the actual location. 
We have put a strong emphasis in the formal treatment of the privacy
guarantees, both in giving a  rigorous definition of  geo-indistinguishability, 
and in providing a mathematical proof that our method  satisfies such property. 
We also have shown how geo-indistinguishability relates to the popular notion of 
differential privacy. 
Finally, we have illustrated the applicability of our method  on a POI-retrieval service, and  we have compared it with other mechanisms in the literature, showing that 
it outperforms those which do not depend on the prior. 

In the future we aim at extending our method to cope with more complex applications, 
possibly involving the sanitization of several (potentially related) locations. 
One important  aspect to consider when generating noise on several data is  the fact that their correlation may 
degrade the level of protection. We aim at devising techniques to control the possible loss of privacy
and to allow the composability of our method.  

\section{Acknowledgements}
This work was partially supported by the European Union 
7th FP under the grant agreement no. 295261 (MEALS), by the projects ANR-11-IS02-0002  LOCALI and ANR-12-IS02-001 PACE, and by the 
INRIA Large Scale Initiative CAPPRIS. The work of Miguel E. Andr\'es was supported by a QUALCOMM grant. The work of Nicol\'as E. Bordenabe was partially funded by the French Defense 
Procurement Agency (DGA) by a PhD grant.
% Kostas: I removed most of the references below. They are papers about hiding
% the identity, not location, so no need to have that many
%\cite{Gruteser:03:MobiSys,
%Gedik:05:ICDCS, Gedik:08:TMC, Mokbel:06:VLDB, Bamba:08:WWW, Pingley:09:ICDCS,
%Manweiler:09:CCS, Hu:09:ICDE, Chow:06:GIS, Ghinita:07:SSTD, Ghinita:07:WWW}

\bibliographystyle{abbrv}
\bibliography{../short} 

%\newpage
\appendix

In this appendix we provide the technical details that have been omitted from the main body of the paper. 

%\subsection{Results from Section~\ref{sec:definitions}}

\remove{
\thmcoincide*
\begin{proof}
	The equivalence of Geo-indistinguishability-I and III can be shown by
	applying Bayes' law. We show here the equivalence between
	Geo-indistinguishability-II and III.

	Assume that $\K$ satisfies \geoind-III. We first show that for all
	$r>0$:
	\begin{align*}
		&P(S | B_r(x)) \\
			%&= \sum_{x'\in \calx} P(S,x' | B_r(x))
			%\\
			&= \sum_{x'\in B_r(x)} P(S,x' | B_r(x))
			  %& P(S,x'|B_r(x))=0, x'\notin B_r(x)
			\\
			&= \sum_{x'\in B_r(x)} P_X(x'|B_r(x))  \K(x')(S)
			\\
			&\ge \sum_{x'\in B_r(x)} P_X (x'|B_r(x)) e^{-\epsilon\,r} \K(x)(S)
				& d(x,x')\le r
			\\
			&= e^{-\epsilon\,r} \K(x)(S)
	\end{align*}
	Then
	\begin{align*}
		P(x|S,B_r(x))
			= \frac{P(S|x)}{P(S|B_r(x))} P(x|B_r(x)) 
			\le  e^{\epsilon\,r} P(x|B_r(x))
	\end{align*}
	For the opposite direction, let $x_1,x_2 \in \calx$, let $r = d(x_1,x_2)$
	and define a prior distribution $P_X^t(x)$ as:
	\[
		P_X^t(x) =
		\begin{cases}
			t & x= x_1 \\
			1-t & x= x_2 \\
			0   & \text{otherwise}
		\end{cases}
	\]
	Using that prior for $t\in(0,1)$ we have for all $S$:
	\begin{align*}
		\K(x_1)(S)
			&= P(S | x_1)
			\\
			&= P(S | x_1, B_r(x_1))
				%& x_1 \in B_r(x_1) \\
			\\
			&= \frac{P(x_1 | S, B_r(x_1))}{P(x_1 | B_r(x_1))} P(S | B_r(x_1)) 
			\\
			&\le e^{\epsilon\,r} P(S | B_r(x_1))
			   %\\ &\quad [\text{hypothesis,}\diam(B_r(x)) = \dx(x,x')]
			\\
			&\le e^{\epsilon\,r} \sum_{x\in\calx} P(S,x | B_r(x_1)) 
			\\
			&\le e^{\epsilon\,r} (t P(S|x_1) + (1-t)P(S|x_2)) 
			\\
			&\le e^{\epsilon\,r} (t \K(x_1)(S) + (1-t) \K(x_2)(S))
				\label{eq2}
	\end{align*}
	Note that we need $t\in (0,1)$ so that $P_X^t(x_1),P_X^t(x_2)$ are positive
	and the conditional probabilities can be defined.
	Finally, taking the $\lim_{t\to 0}$ on both sides of the above inequality we get
	$
		\K(x_1)(S) \le e^{\epsilon\,r} \K(x_2)(S)
	$
\end{proof}

%\remove{
\thmcomposition*
\begin{proof}
	Let $\xb=(x_1,\dots,x_n),\xb'=(x_1',\ldots,x_n')$ such that
	$\dmax(\xb,\xb')\le r$. This implies that $d(x_i,x_i')\le r$, $1\le i\le n$.
	We have:
	\begin{align*}
		P(\zb|\xb)
		&= \smallprod{i} P(z_i|x_i) \\
		&\le \smallprod{i} e^{\epsilon r} P(z_i|x_i') \\
		&= e^{n\epsilon r} \smallprod{i} P(z_i|x_i') \\
		&= e^{n\epsilon r} P(\zb|\xb')
	\end{align*}
\end{proof}
%}
} %remove

\remove{  %%%%% begin remove
\subsection{Coefficient of the continuous planar laplacian}\label{app:normalization factor}
The definition of the pdf of the planar laplacian is: 
\[
D_\epsilon (\vect{x}_0) (\vect{x}) = \lambda \,  e^{-\epsilon \,{d(\vect{x}_0,\vect{x})}}
\]
In order to determine $\lambda$ we require that the integration of the above pdf on the whole $\reals^2$ gives 1. 
For simplicity, and without loss of generality, we assume that $\vect{x}_0=(0,0)$. 
Thus we obtain the following constraint:
\[
%\begin{equation}\label{eqn:Laplacian integration}
\int\!\!\!\!\int_{\reals^2} \lambda_\epsilon e^{-\epsilon\; \sqrt{x^2+y^2}} dx\  dy  = 1
%\end{equation}
\]
In order to solve this equation it is convenient to apply a transformation into the polar system. Using the standard formula to transform a pdf from cartesian coordinates to polar coordinates $(r,\theta)$,  we obtain
\[
%\begin{equation}\label{eqn:Laplacian integration polar}
\int_0^{2\pi}\left(\int_{0}^\infty \lambda \,   r \, e^{-\epsilon\,{r}} dr\right)  d\theta  = 1
%\end{equation}
\]
Using the integration by parts, we derive
\[
\int_0^{2\pi}\lambda \left.\left(- \frac{r\, e^{-\epsilon\,r}}{\epsilon} - \frac{e^{-\epsilon\,r}}{\epsilon^2} \right)\right|_0^\infty  d\theta  = 1
\]
namely:
\[
\int_0^{2\pi}\frac{\lambda}{\epsilon^2} \; d\theta  = 1
\]
from which we conclude
\[
\lambda = \nicefrac{\epsilon^2}{2\,\pi}
\]
\hfill$\qed$
%from which we conclude:
%\begin{equation}\label{eqn:polar laplacian}
%D_\epsilon  (r,\theta) =  \frac{\epsilon^2}{2\,\pi} \, r\, e^{-\epsilon \,r}
%\end{equation}

\subsection{The planar laplacian satisfies   geo-indistinguishability}\label{app:geo-indistinguishability}
Given the definition of $D_\epsilon (\vect{x}_0) (\vect{x})$ in (\ref{eqn:planar laplacian}), by triangular inequality we have 
\[ %\begin{equation}\label{eqn:pdf}
D_\epsilon (\vect{x}_0) (\vect{x})\leq e^{\epsilon \,d(\vect{x}_0,\vect{x}'_0)}{D_\epsilon (\vect{x}'_0) (\vect{x}) }
\] %\end{equation}
Using well-known properties of integrals, we derive
\[
\int_S D_\epsilon (\vect{x}_0) (\vect{x})  ds \leq \int_S e^{\epsilon \,d(\vect{x}_0,\vect{x}'_0)} \, {D_\epsilon (\vect{x}'_0) (\vect{x}) }ds
\]
and 
\[
\int_S D_\epsilon (\vect{x}_0) (\vect{x})  ds \leq e^{\epsilon \,d(\vect{x}_0,\vect{x}'_0)}\int_S {D_\epsilon (\vect{x}'_0) (\vect{x}) }ds
\]
Now, taking into account the definition of $\K$: 
\[
{\K}(\vect{x}_0)(S)  = \int_S D_\epsilon (\vect{x}_0) (\vect{x})  ds
\]
we derive
\[
{\K}(\vect{x}_0)(S)  \leq e^{\epsilon \,d(\vect{x}_0,\vect{x}'_0)}{\K}(\vect{x}'_0)(S) 
\]
 \hfill$\qed$

%\subsection{The planar Laplacian satisfies \geoind}

%When each distribution $K$ is given by a probability density function
%$D(x)$, we can show that the mechanism satisfies \geoind{} by showing that the density
%function pointwise satisfies $D(x)(z) \le e^{\epsilon r} D(x')(z)$  for
%$d(x,x')\le r$. Let $x,x',z \in\reals^2, d(x,x')\le r$, we have
%\begin{align*}
%	D(x)(z)
%		&= \frac{\epsilon^2}{2\pi} e^{-\epsilon \,d(x,z)} \\
%		&\le \frac{\epsilon^2}{2\pi} e^{-\epsilon \, (d(x',z) - d(x,x') )}
%			&\text{triangular ineq.} \\
%		 &\le e^{\epsilon r} \frac{\epsilon^2}{2\pi} e^{-\epsilon \, d(x',z)}
%			& d(x,x')\le r \\
%		 &= e^{\epsilon r} D(x')(z)
%\end{align*}

} %%%%% end remove

%\subsection{The discretization  preserves   geo-in\-dis\-tin\-guisha\-bil\-ity}\label{app:geo-indistinguishability discrete}
\thmgeodp*
\begin{proof}
        The case in which $\vect{x}_0 = \vect{x}'_0$ is trivial. We consider therefore only the case in which $\vect{x}_0 \neq \vect{x}'_0$. 
        Note that in this case $d(\vect{x}_0,\vect{x}'_0) \geq u$. 
	We proceed by determining  an upper bound on ${\K}_{\epsilon'}(\vect{x}_0)(\vect{x})$ and  a lover bound on 
	${\K}_{\epsilon'}(\vect{x}'_0)(\vect{x})$
	for generic
	$\vect{x}_0$, $\vect{x}'_0$ and $\vect{x}$ such that $d(\vect{x}_0,\vect{x}), d(\vect{x}'_0,\vect{x})\le r_{\max}$. 
	Let $S$ be the set of points for which $x$ is the closest point in $\grid$, namely:
	\[S = R(\vect{x}) = \{\vect{y}\in\reals^2\, |\, \forall \vect{x}'\in {\cal G}. \; d(\vect{y},\vect{x}') \leq d(\vect{y},\vect{x}')\}\]
	 Ideally, the points remapped in $\vect{x}$
	would be exactly those in $S$. 
	However, due to the finite precision of the machine, the points actually remapped in $\vect{x}$ are those of $R_{\cal W}(\vect{x})$ 
	(see Section~\ref{sec:discrete}).
	Hence the probability of $\vect{x}$ is that of $S$ plus or minus the small rectangles\footnote{$W$ is actually a fragment of a circular crown, but    
	since  $\delta_\theta$ is very small, it approximates a rectangle. Also, the side of $W$ is not exactly $r\,\delta_\theta$, it is a number in the interval 
	$[(r-\nicefrac{u}{\sqrt{2}})\,\delta_{\theta},(r+\nicefrac{u}{\sqrt{2}})\,\delta_{\theta}]$. However $\nicefrac{u}{\sqrt{2}}\,\delta_\theta$ is very small with 
	respect to the other quantities involved, hence we consider negligible this difference.}
	 $W$ of size $\delta_r \times r\, \delta_{\theta}$ at the border of $S$, where $r= d(\vect{x}_0,\vect{x})$, see Figure \ref{fig:rectangle}.
	\begin{figure}[t]%
			\centering
			\includegraphics[width=0.6\columnwidth]{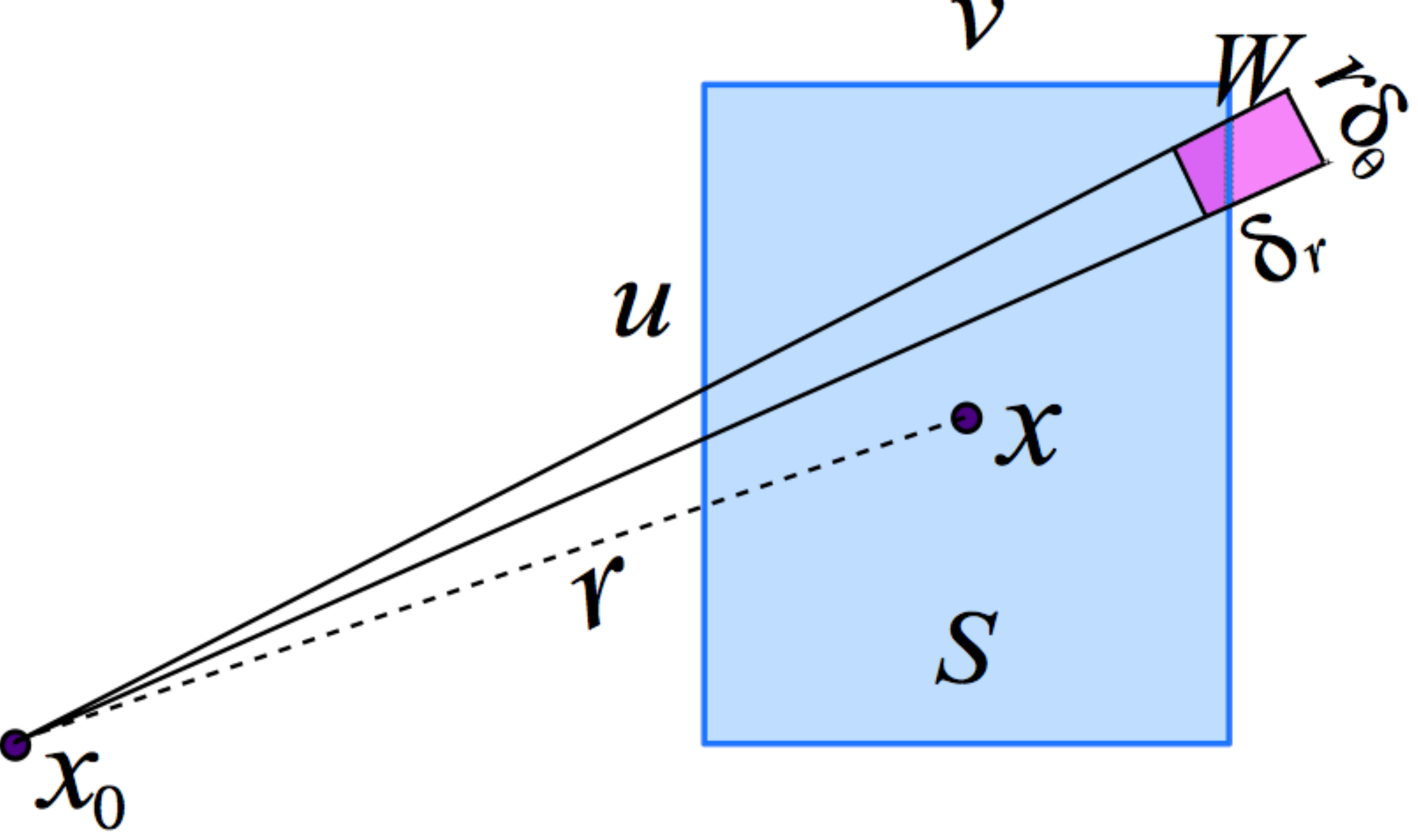}%
			\caption{Bounding the probability of $\vect{x}$ in the discrete Laplacian.}%
			\label{fig:rectangle}%
	\end{figure}
	Let us denote by $S_W$ the total  area of these small rectangles  $W$ on one of the sides of $S$. 
	Since $d(\vect{x}_0,\vect{x})\le r_{\max}< \nicefrac{u}{\delta_\theta}$, and $\delta_r<r_{\max}\delta_\theta$, we have that $S_W$ is less than $\nicefrac{1}{q}$ of the area of 
	$S$, where $q=  \nicefrac{u}{r_{\max} \delta_{\theta}}$. %Since the area of $S$ is $u\,v$, we have $S<\nicefrac{u\,v}{q}$.		
	The probability density on this area differs at most by a factor $e^{\epsilon'u}$ from that of the other points in $S$. 
	%where  $r=d(\vect{x}_0,\vect{x})$. 
	Finally, note that on two sides of $S$ the rectangles $W$ contribute positively to ${\K}_{\epsilon'}(\vect{x}_0)(\vect{x})$, while on two sides they 
	contribute negatively. Summarizing, we have: 
	\begin{equation}\label{eqn:bound1}
	{\K}_{\epsilon'}(\vect{x}_0)(\vect{x})\;\leq\; (1+\frac{2\,e^{\epsilon'u}}{q}) \int_S D_{\epsilon'} (\vect{x}_0) (\vect{x}_1)  ds
	\end{equation}
	and
	\begin{equation}\label{eqn:bound2}
	(1-\frac{2\,e^{\epsilon'u}}{q}) \int_S D_{\epsilon'} (\vect{x}'_0) (\vect{x}_1)  ds  \; \leq \; {\K}_{\epsilon'}(\vect{x}'_0)(\vect{x})
	\end{equation}
	Observe now that 
	\[\frac{D_{\epsilon'} (\vect{x}_0) (\vect{x}_1)}{D_{\epsilon'} (\vect{x}'_0) (\vect{x}_1) }= e^{-{\epsilon'} (d(\vect{x}_0,\vect{x}_1)- d(\vect{x}'_0,\vect{x}_1))}\]
	By triangular inequality we obtain
	\[ %\begin{equation}\label{eqn:pdf}
	D_{\epsilon'} (\vect{x}_0) (\vect{x}_1)\leq e^{\epsilon' \,d(\vect{x}_0,\vect{x}'_0)}{D_{\epsilon'} (\vect{x}'_0) (\vect{x}_1) }
	\] %\end{equation}
	from which  we derive
	\begin{equation}\label{eqn:dise}
	\int_S D_{\epsilon'} (\vect{x}_0) (\vect{x}_1)  ds \leq e^{\epsilon' \,d(\vect{x}_0,\vect{x}'_0)}\int_S {D_{\epsilon'} (\vect{x}'_0) (\vect{x}_1) }ds
	\end{equation}
	from which, using \eqref{eqn:bound1},  \eqref{eqn:dise}, and  \eqref{eqn:bound2},
	we obtain
	\begin{equation}\label{eqn:almost dp}
	{\K}_{\epsilon'}(\vect{x}_0)(\vect{x})  \leq 
	e^{\epsilon' \,d(\vect{x}_0,\vect{x}'_0)}\,{\K}_{\epsilon'}(\vect{x}'_0)(\vect{x})\,\frac{q+2\,e^{\epsilon'u}}{q -2\,e^{\epsilon'u}}
	\end{equation}
	Assume now that
	\[\epsilon' + \frac{1}{u} \ln \, \frac{q + 2\, e^{\epsilon'u}}{q - 2\, e^{\epsilon'u}}\leq \epsilon\]
	Since we are assuming $d(\vect{x}_0,\vect{x}'_0)\geq u$, we derive:
	\begin{equation}\label{eqn:almost there}
	e^{\epsilon' \,d(\vect{x}_0,\vect{x}'_0)}\,\frac{q+2\,e^{\epsilon'u}}{q -2\,e^{\epsilon'u}}\le 
	e^{\epsilon \,d(\vect{x}_0,\vect{x}'_0)}
	\end{equation}
	Finally, from \eqref{eqn:almost dp} and \eqref{eqn:almost there}, we conclude.  
\end{proof}

%\subsection{The truncation  preserves   geo-in\-dis\-tin\-guisha\-bil\-ity}\label{app:geo-indistinguishability truncated}
\thmgeotruncateddp*
\begin{proof}
	 The proof proceeds like the one for Theorem~\ref{theo:geo-dp}, except when  $R(\vect{x})$ is on the border of $\cala$. In this latter case, 
	the probability on $\vect{x}$ is given not only by the probability on $R(\vect{x})$ (plus or minus the small rectangles $W$ 
	-- see the proof of Theorem~\ref{theo:geo-dp}), but also by the probability of  the part $C$ 
	of the cone determined by  $\vect{o}$, $R(\vect{x})$, and  lying outside 
	$\cala$ (see Figure~\ref{fig:circle}). 
 \begin{figure}[t]%
		\centering
		\includegraphics[width=0.5\columnwidth]{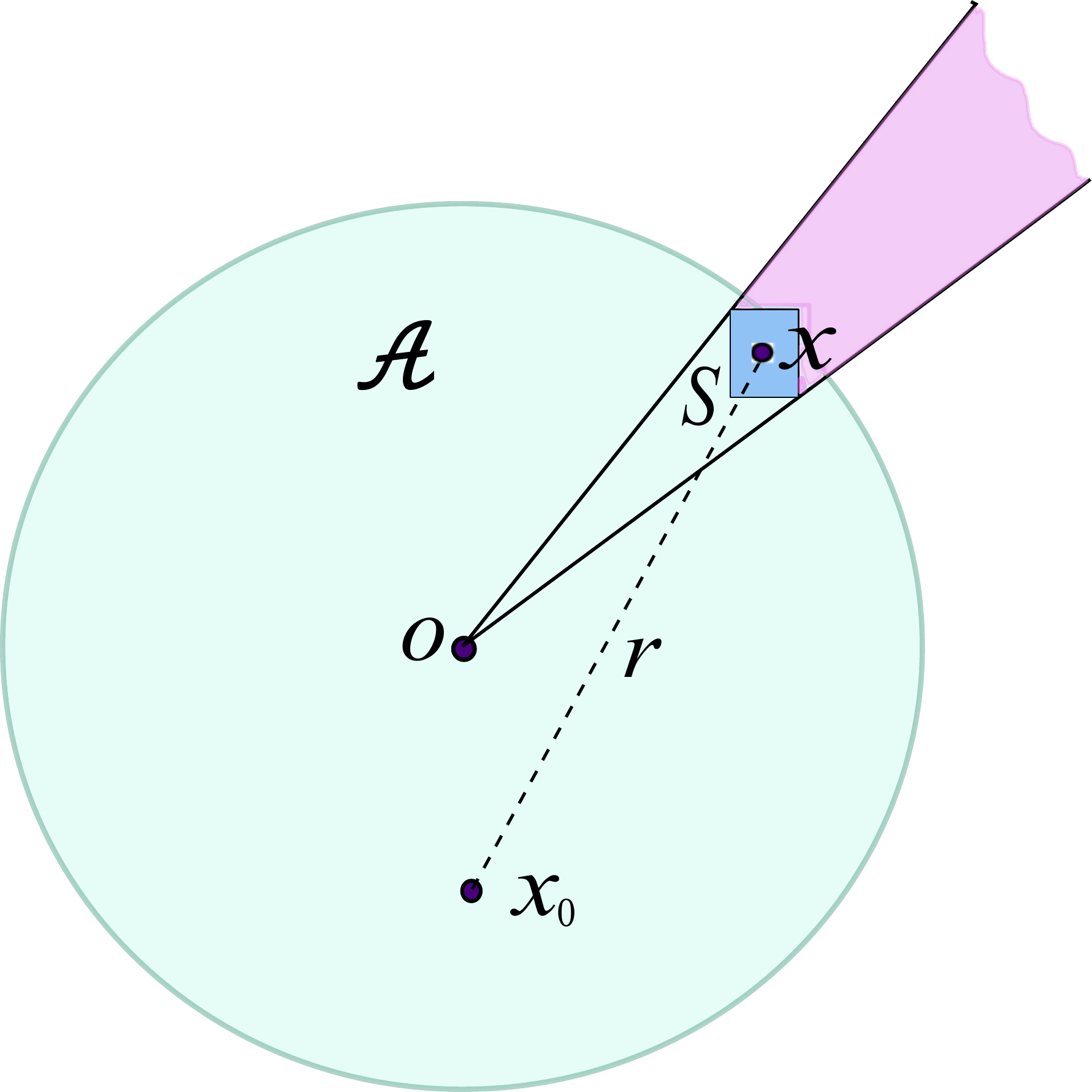}\\[1ex]%
		\caption{Probability of $\vect{x}$ in the truncated discrete laplacian.}%
		\label{fig:circle}%
\end{figure}
Following a similar reasoning as in the proof of Theorem~\ref{theo:geo-dp} we get
\[	
	{\K}^T_{\epsilon'}(\vect{x}_0)(\vect{x})\;\leq\; (1+\frac{2\,e^{\epsilon'u}}{q}) \int_{S\cup C} D_{\epsilon'} (\vect{x}_0) (\vect{x}_1)  ds
\]
	and
\[
	(1-\frac{2\,e^{\epsilon'u}}{q}) \int_{S\cup C} D_{\epsilon'} (\vect{x}'_0) (\vect{x}_1)  ds  \; \leq \; {\K}^T_{\epsilon'}(\vect{x}'_0)(\vect{x})
\]
	The rest follows as in the proof of Theorem~\ref{theo:geo-dp}.
\end{proof}

\remove{
\subsection{Study of the extra noise required by  the discretization}
\label{app:relation}
Figure~\ref{fig:epsilon_prime} shows the relation between $\epsilon$ and  the maximal $\epsilon'$ satisfying the condition of Theorem~\ref{theo:geo-dp}. In all cases the grid unit is $u = 3\cdot 10^{-3}$ Km  $= 3$ m, and the other parameters are as follows:  
\begin{itemize}
\item 
The green line corresponds to $q = 3\cdot 10^{9}$. For instance this value 
can be obtained with double precision ($16$ significant digits, i.e., $\delta_\theta = 10^{-16}$) and $r_{\max} = 10^4$ Km.  In the case of double precision, even for  much larger values of  $r_{\max}$ (up to about $10^6$ Km)  $\epsilon'$ coincides with $\epsilon$. 
\item
The magenta line corresponds to $q = 3\cdot 10^{2}$. This value can be obtained with single precision ($7$ significant digits, i.e., $\delta_\theta = 10^{-7}$) and $r_{\max} = 10^2$  Km.  In this case we cannot go much higher for $r_{\max}$ without $\epsilon'$ diverging dramatically from $\epsilon$. Furthermore, the smallest possible value  for $\epsilon$ is about $4.5$, which means that at most we can ensure $4.5$-geo-indistinguishability. 
\item
The blue line corresponds to $q = 3\cdot 10^{3}$, which can still be obtained  with single precision at the price of   reducing previous $r_{\max}$ by a factor $10$ ($r_{\max} = 10$ Km). Alternatively we could obtain this value by increasing both the precision and $r_{\max}$: For instance,  with an intermediate precision of $9$ significant digits ($\delta_\theta = 10^{-9}$) and $r_{\max} = 10^3$ Km. 
\end{itemize}

% (The graph does not depend on the unit of measurement of $\delta_r$ and $u$, but they have to be the same unit,  
% and that of  $\epsilon' $ and $\epsilon$ has to be the inverse, as usual.) 
% As we can see, the amount of additional noise is negligible in case of double precision, even for large values of $r_{\max}$, while it becomes noticeable in single precision, in which case it is necessary to reduce $r_{\max}$  to contain the effect. 

\begin{figure}[tb]%
		\centering
		\includegraphics[width=0.8\columnwidth]{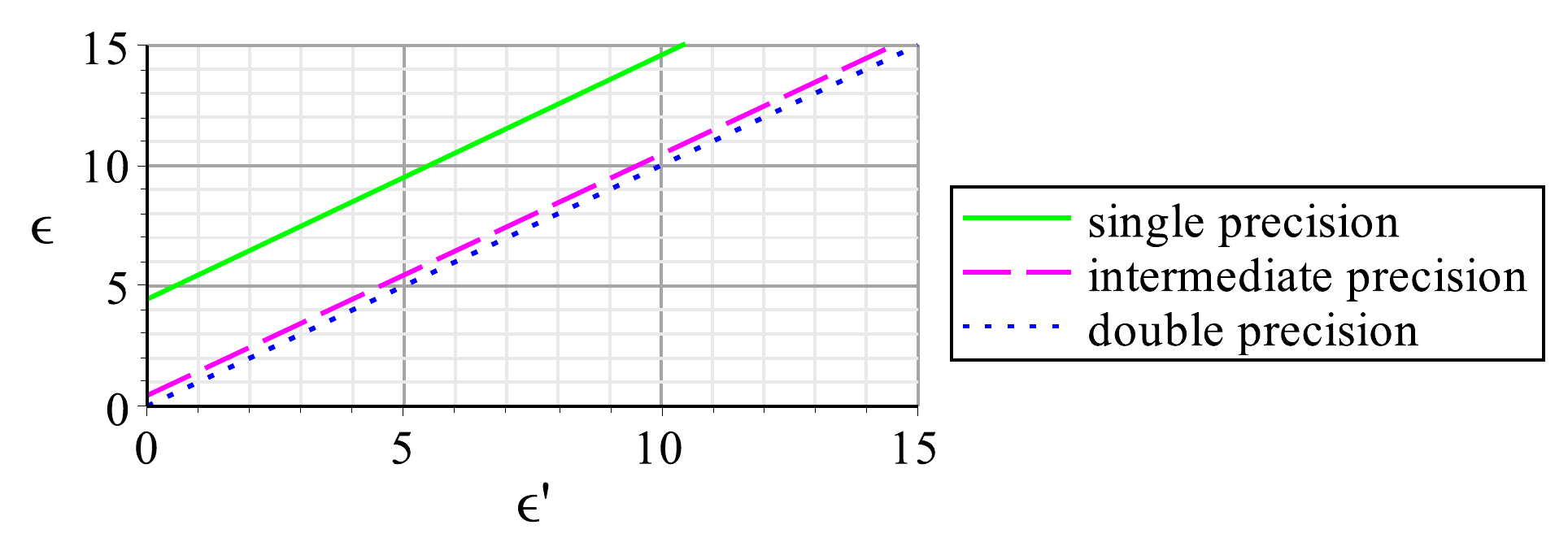}\\[-1ex]%
		\caption{The relation between $\epsilon$ and  $\epsilon'$ for various precisions.  }%
		\label{fig:epsilon_prime}%
		\vspace{-12pt}
\end{figure}

\subsection{Experiments on the LODES dataset}\label{app:experiments}

In order to evaluate the accuracy of the sanitized dataset generated by our algorithm (and thus of our algorithm as a data sanitizer) 
we implemented our perturbation mechanism and conducted a series of experiments 
% varying the privacy parameter $\epsilon$ of the algorithm. 
focusing on the ``home-work commute distance'' analysis provided by the OnTheMap application. 
This analysis provides, for a given area (specified as, say, state or county code), 
a histogram classifying the individuals in the dataset residing in the given area according to the distance between their residence location and their work location. The generated histogram contains four buckets representing different ranges of distance: (1) from zero to ten miles, (2) from ten to twenty five miles, (3) from twenty five to fifty miles, and (4) more than fifty miles.%\footnote{Here we choose miles as unit of measure, in order to compare our results with  the literature  and with online tools about the LODES dataset.}. 

We have chosen the San Francisco (SF) County as residence area for our experimental analysis. Additionally, we restrict the work location of individuals residing in the San Francisco county to the state of California. The total number of individuals satisfying these conditions amounts to 374.390. All experiments have been carried on using version 6.0 of the LODES dataset. In addition, the mapping from census blocks to their corresponding centroids has been done using the 2011 TIGER census block shapefile information provided by the Census Bureau. 

We now proceed to compare the LODES dataset -- seen as a histogram -- with several sanitized versions of it generated by our algorithm.  Figure \ref{fig:ODbyEpsilon_cat} (a) depicts how the geographical information degrades when fixing $r$ to $1.22$ miles (so to ensure \geoind{} within $10\%$ of the land area of the SF County) and varying $\ell$. The precision parameters were chosen as follows: $u=10^{-3}$ miles, $\cal A$'s diameter was set to $10^4$ miles, and the standard double precision values for $\delta_r$ and $\delta_\theta$ (for the corresponding ranges).  
%The San Francisco county has a land area of approximately 46,87 sq mi (~7x7)121 km$^2$ and 
%We first note that, as our algorithm takes as input a dataset with home and work coordinates, the first step is to map pairs of census blocks $(hBlock, wBlock)$ into pairs of coordinates $(hCoord, work\_coord)$. For this purpose we used the 2011 TIGER/line Shapefiles of each census block provided by the census bureau. These shapefiles contain detailed geographic information about each census block. In our experiments we use the geographic centroid of each census block to represent the geographic location of the block.

%In these experiments, we have followed the approach of deriving $\epsilon'$ from the lever of privacy $\ell$ and  the radius of protection $r$, as illustrated in the introduction. We recall that $\epsilon'=\nicefrac{\ell}{r}$.

%  \begin{figure}[h!]
%      \centering
%      \includegraphics[width=8.5cm]{figures/histograms-1legend-center.pdf}
%%      \includegraphics[width=6cm]{figures/ODbyEpsilon_cat1.pdf}
%   \caption{Home-work commute distance  for  $r\!=1.22$ and various $\ell$.}\label{fig:ODbyEpsilon_cat}
% \end{figure}

%
% \begin{figure}[h!]
%      \centering
%      \includegraphics[width=8.5cm]{figures/histograms-2legends.pdf}
%%      \includegraphics[width=6cm]{figures/ODbyEpsilon_cat1.pdf}
%   \caption{Home-work commute distance  for  $r\!=1.22$ and various $\ell$.}\label{fig:ODbyEpsilon_cat}
% \end{figure}

We have also conducted experiments varying $r$ and fixing $\ell$. For instance, if we want to  provide  geo-indistinguishability for $5\%$, $10\%$, and $25\%$ of the land area of the SF county (approx. $46.87$ mi$^2$), we can set $r\!=\!0.86$,  $1.22$, and $1.93$ miles, respectively. Then by taking $\ell=\ln(2)$ we get an histogram very similar to the previous one. This is not surprising as the noise generated by our algorithm  depends only on the ratios $\nicefrac{\ell}{r}$, which are similar for the  values above. 

\begin{figure}[t]
      \centering
      \includegraphics[width=8.5cm]{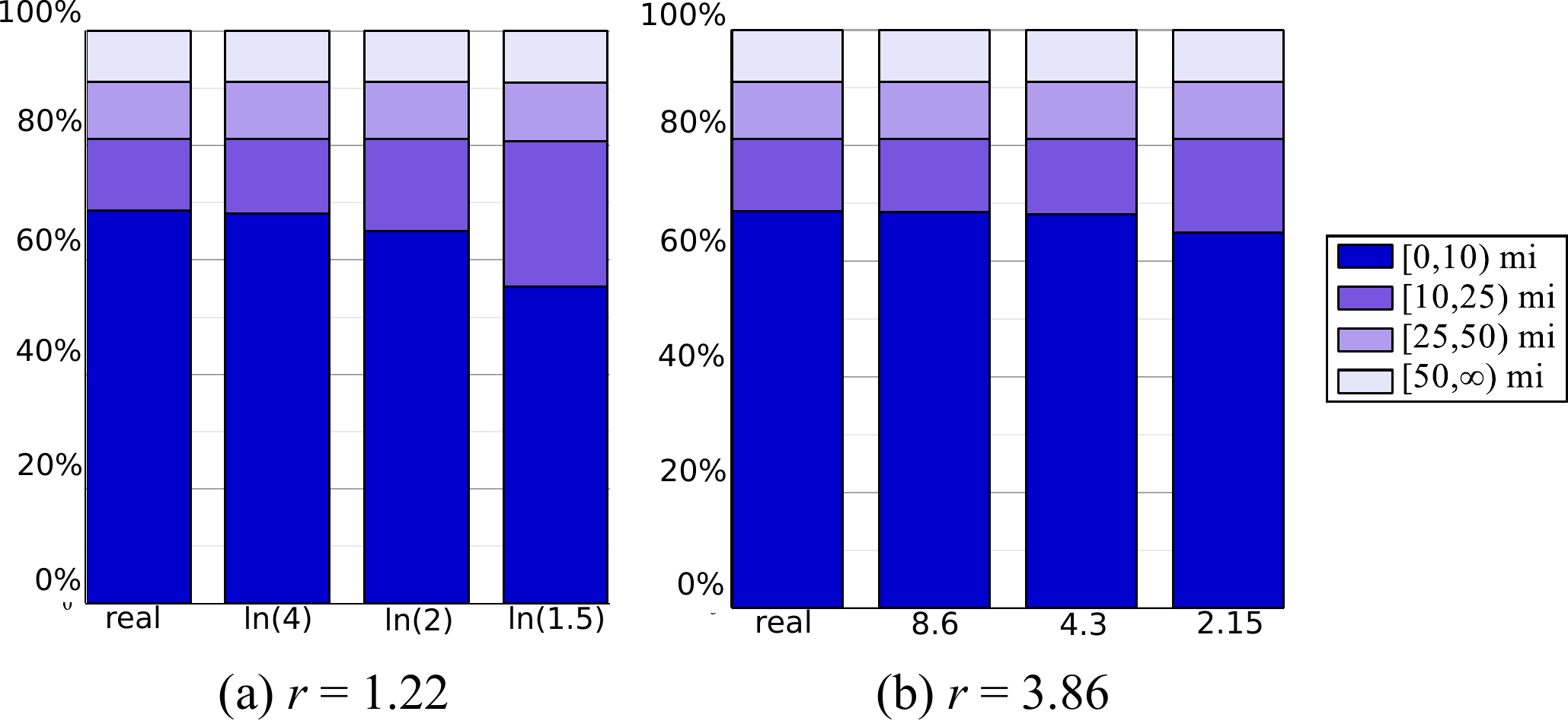}
   \caption{Home-work commute distance for various levels $\ell$.}\label{fig:ODbyEpsilon_cat}
   \vspace{-10pt}
 \end{figure}

% \begin{figure}[h!]
 %     \centering
%      \includegraphics[width=6cm]{figures/ODbyArea_cat}
%   \caption{(b) Home-work commute distance histograms for several algorithm's outputs with $\epsilon=\ln(2)$ and various values of $r$}\label{fig:ODbyArea_cat}
 %\end{figure}

As shown in Figure \ref{fig:ODbyEpsilon_cat} (a), our algorithm has little effect on the bucket counts corresponding to mid/long distance commutes: over twenty five miles the counts of the sanitized dataset are almost identical to those of the input dataset -- even for the higher degrees of privacy. For short commutes on the other hand,  the increase in privacy degrades the accuracy of the sanitized dataset: several of the commutes that fall in the 0-to-10-miles bucket in the original data fall instead in the 10-to-25-miles bucket in the sanitized data.
% (due to the perturbation on the home location coordinates). 

%The degradation of sanitized data (for short commutes in this particular case) was however to be expected, as there is a trade-off between the privacy guarantees provided by the sanitizing  algorithm and the utility of the resulting sanitized data. 
%\footnote{In order to further asses the quality of the sanitized data obtained by our algorithm it would be ideal to be able to compare 
%the quality of the sanitized dataset obtained by our algorithm 
%them with a sanitized dataset obtained using the synthetization techniques employed by the Census Bureau. 
%However, so far it has not been possible for the authors to obtain any test data from the Census Bureau. 
%At the moment we are in contact with the Census Bureau office responsible for the LODES dataset and expect to be able to obtain test data before submitting the final version of this article.

After analyzing the accuracy of the sanitized datasets produced by our algorithm for several levels of privacy, we proceed to compare our approach with the one followed by the Census Bureau to sanitize the LODES dataset. Such comparison is unfortunately not straightforward; on the one hand, the approaches provide different privacy guarantees (see discussion below) and, on the other hand, the Census Bureau is not able to provide us with a (sanitized) dataset sample produced by their algorithm (which would allow us to compare both approaches in terms of accuracy) as this might compromise the protection of the real data.

The algorithm used by the Census Bureau satisfies a notion of privacy that called $(\epsilon, \delta)$-probabilistic differential privacy, which is a relaxation of standard differential privacy that provides \edpold{} with probability at least $1-\delta$ \cite{Machanavajjhala:08:ICDE}. In particular, their algorithm satisfies $(8.6 , 0.00001)\allowbreak$-probabilistic differential privacy. This level of privacy could be compared to geo-indistinguishability for $\ell\!=\!8.6$ and $r\!=\!3.86$, which corresponds to providing protection in an area of the size of the SF County.
%It is not possible to compare these results with those obtained for our algorithm because they refer to different privacy notions. However, their $\epsilon$ corresponds to our $\ell$, and  we note that for such high value of $\ell$ (i.e. $8.6$),  with our  method we can offer the user geo-indistinguishability within the whole San Francisco county ($r=3.86$ miles) while still providing useful results. 
Figure  \ref{fig:ODbyEpsilon_cat} (b) presents the results of our algorithm for such level of privacy and also for higher levels.

% \begin{figure}[h!]
%      \centering
%      \includegraphics[width=6cm]{figures/ODwithEpsilon86_cat1.pdf}
%   \caption{Home-work commute distance for $r\!=3.86$, corresponding to the  San Francisco county land area, and various (high) values of $\ell$.}\label{fig:ODwithEpsilon86}.
% \end{figure}

It becomes clear that, by allowing high values for $\ell$ ($\ell=8.6=\ln(5432)$, $\ell=4.3=\ln(74)$, and $\ell=2.15=\ln(9)$)
%-- compare to the values 
%$\ln(4)$, $\ln(2)$, and $\ln(1.5)$ used in our example above) 
it is possible to provide privacy in large areas without significantly diminishing the quality of the  sanitized dataset. 
}

\end{document}